\documentclass{article}
\usepackage{amssymb}
\usepackage{amsmath}
\usepackage{amsthm}

\input amssym.def
\input amssym

\newtheorem{theorem}{Theorem}
\newtheorem{lemma}{Lemma}
\newtheorem{cor}{Corollary}

\theoremstyle{definition}
\newtheorem{defin}{Definition}

\newcommand{\hi}{\mathcal{H}} %Hilbert space
\newcommand{\tsh}{\mathcal{T}_s(\hi)} %self-adjoint trace-class operators
\newcommand{\bsh}{\mathcal{B}_s(\hi)} %bounded selfadjoint operators
\newcommand{\sh}{\mathcal{S}(\hi)} %quantum states
\newcommand{\eh}{\mathcal{E}(\hi)} %quantum effects
\newcommand{\os}{(\Omega,\Sigma)} %measurable space
\newcommand{\mos}{\mathcal{M}_{\mathbb R}\os} %real measures
\newcommand{\fos}{\mathcal{F}_{\mathbb R}\os} %real measurable functions
\newcommand{\sos}{\mathcal{S}\os} %classical states
\newcommand{\eos}{\mathcal{E}\os} %classical effects

\newcommand{\fii}{\varphi}

\newcommand{\n}{\| \, . \, \|}

\newcommand{\sigtop}[1]{\sigma({#1})} %sigma-topology
\newcommand{\dual}[2]{\langle {#1},{#2}\rangle} %dual product
\newcommand{\tr}[1]{\mathrm{tr} \, {#1}} %trace
\newcommand{\no}[1]{\left\|{#1}\right\|} %norm
\newcommand{\kb}[2]{|#1\,\rangle\langle\,#2|} %ketbra
\newcommand{\delo}{\delta_{\omega}} %point measure
\begin{document}

\title{From the Attempt of Certain Classical Reformulations of Quantum Mechanics to Quasi-Probability Representations}

\author{Werner Stulpe\thanks{Electronic mail: stulpe@fh-aachen.de}\\
{\small Aachen University of Applied Sciences, J\"ulich Campus, D-52428, Germany}}

\date{}
\maketitle

\begin{abstract}
\noindent The concept of an injective affine embedding of the quantum states into a set of classical states,
i.e., into the set of the probability measures on some measurable space, as well as its relation to
statistically complete observables is revisited, and its limitation in view of a classical reformulation
of the statistical scheme of quantum mechanics is discussed. In particular, on the basis of a theorem concerning
a non-denseness property of a set of coexistent effects, it is shown that an injective classical embedding
of the quantum states cannot be supplemented by an at least approximate classical description of the quantum
mechanical effects. As an alternative approach, the concept of quasi-probability representations
of quantum mechanics is considered.\\
\mbox{}
Running title: Quantum probability in terms of measures and functions.
\end{abstract}

\section{Introduction}

The problem of describing quantum states by (quasi-) probability densities or (quasi-) probability measures
was originated by Wigner's famous paper \cite{wig32} and has been discussed and investigated
by many other authors. One approach of particular interest is that of injective affine representations of
the density operators by probability densities on phase space \cite{ali77a;77b,pru84,sch82,stu92} or, more generally,
by probability measures on some measurable space \cite{bug93,bus04,sin92,stu94;98,stu97}. This procedure
is related to the so-called informationally complete positive operator-valued measures
(informationally complete POVMs) or, in the author's terminology, to the statistically complete observables
\cite{ali77a;77b,bus95,bus89,pru84,sin92,stu97}. Another very interesting approach is given by
a surjective affine mapping from the probability measures on the projective Hilbert space onto the set
of all density operators, thus describing the density operators by equivalence classes of probability measures
on the projective Hilbert space \cite{bel95,bug91;93,bug93,bus04,stu08,stu01}.

The approach related to informationally complete POVMs is currently still a subject of several articles
(e.g., \cite{fil10,fuc02}), however, initiated by quantum information, in most cases
a finite-dimensional Hilbert space is presupposed. In this paper, we revisit that approach
to classical reformulations of the probabilistic frame of quantum mechanics where the main focus
is on infinite-dimensional Hilbert spaces. In the context of such a classical reformulation
the quantum states and effects are represented by probability measures and functions on some measurable space,
respectively. As a central point, we present a theorem on a non-denseness property of a set of coexistent effects;
by means of that property, we show an essential limitation of the considered classical reformulations of
quantum mechanics.

We briefly describe the mathematical framework we are working within and fix our notation. Let a complex separable
Hilbert space $ \hi \neq \{0\} $ be given. We denote the real vector space of the self-adjoint trace-class operators
acting in $ \hi $ by $ \tsh $ and the convex set of the positive trace-class operators of trace $1$ by $ \sh $;
the operators of $ \sh $ are the density operators and describe the quantum states. The pair $ (\tsh,\sh) $
is a base-normed Banach space with closed positive cone, the base norm being the trace norm. We denote
the real vector space of all bounded self-adjoint operators acting in $ \hi $ by $ \bsh $ and the unit operator by
$I$. The pair $ (\bsh,I) $ where $ \bsh $ is equipped with its order relation, is an order-unit normed Banach space
with closed positive cone, the order-unit norm being the usual operator norm. The elements of the order-unit interval
$ \eh := [0,I] $ describe the quantum mechanical effects, i.e., the quantum measurements with only two possible
outcomes, say, `yes' and `no'. As is well known, $ \bsh $ can be considered as the dual space $ (\tsh)' $ where
the duality is given by the trace functional
\[
(V,A) \mapsto \dual{V}{A} := \tr{VA},
\]
$ V \in \tsh $, $ A \in \bsh $. The restriction of this bilinear functional to $ \sh \times \eh $ is
the quantum probability functional; $ \tr{WA} $ is the probability for the outcome `yes' of the effect
$ A \in \eh $ in the state $ W \in \sh $. Thus, $ \dual{\tsh}{\bsh} $ is a dual pair of vector spaces
(in fact a statistical duality \cite{sin92,wer83}) which encompasses the probabilistic structure
of usual quantum mechanics \cite{bus95,dav76,dav70,hol01,lud83}.

We recall that the extreme points of the convex set $ \sh $, i.e., the pure quantum states, are the one-dimensional
orthogonal projections $ P_{\fii}: = \kb{\fii}{\fii} $, $ \fii \in \hi $, $ \no{\fii} = 1 $. The extreme points of
the convex set $ \eh $ are all orthogonal projections of $ \hi $, these are sometimes called {\it sharp} effects
whereas the other effects are called {\it unsharp}.---We also recall that $ \sigtop{\bsh,\tsh} $ is
the weak-* Banach-space topology of $ \bsh $, i.e., the coarsest topology on $ \bsh $ in which
the elements of $ \tsh $, considered as linear functionals on $ \bsh $, are continuous. We call
this topology briefly the \emph{$ \sigma $-topology}; it is the analog of the ultraweak operator topology
$ \sigma(\mathcal{B}(\hi),\mathcal{T}(\hi)) $ which is defined for the corresponding complex vector spaces
of operators whose elements are not necessarily self-adjoint.

For a general measurable space $ \os $ where $ \Omega $ is a nonempty set and $ \Sigma $ an arbitrary
$ \sigma $-algebra of subsets of $ \Omega $, let $ \mos $ be the real vector space of the real-valued measures on
$ \os $ (i.e., of the $ \sigma $-additive real-valued set functions on $ \Sigma $). As a consequence of the
$ \sigma $-additivity, these set functions are bounded for which reason there are also called the bounded
signed measures on $ \os $. We denote the convex subset of the normalized positive measures by
$ \sos $; the elements of $ \sos $ are probability measures and describe classical states. The pair
$ (\mos,\sos) $ is a base-normed Banach space with closed positive cone, the base norm being
the total-variation norm. By $ \fos $ we denote the real vector space of the real bounded $ \Sigma $-measurable
functions on $ \Omega $ and by $ \chi_B $ the characteristic function of a set $ B \in \Sigma $. The pair
$ (\fos,\chi_{\Omega}) $ together with the order relation of $ \fos $ is an order-unit normed Banach space
with closed positive cone, the order-unit norm being the supremum norm. The elements of the order-unit interval
$ \eos := [0,\chi_{\Omega}] $ describe the classical effects. By the bilinear functional given by the integral
\[
(\nu,f) \mapsto \dual{\nu}{f} := \int f \, \mathrm{d}\nu,
\]
$ \nu \in \mos $, $ f \in \fos $, the spaces $ \mos $ and $ \fos $ are placed in duality to each other; in particular,
$ \fos $ can be considered as a norm-closed subspace of the dual space $ (\mos)' $. The space $ \fos $ is in general
smaller than $ (\mos)' $, but its elements separate the elements of $ \mos $. The restriction of
$ (\nu,f) \mapsto \dual{\nu}{f} $ to $ \sos \times \eos $ is the classical probability functional;
$ \int f \, \mathrm{d}\mu $ is the probability for the outcome `yes' of the effect $ f \in \eos $ in the state
$ \mu \in \sos $. Again, $ \dual{\mos}{\fos} $ is a dual pair of vector spaces (also a statistical duality),
$ \dual{\mos}{\fos} $ encompassing classical probability theory \cite{bug96,bug98,dav70,gud98,sin92,stu86}.

We remark that the Dirac measures $ \delo $, $ \omega \in \Omega $, are extreme points of the convex set $ \sos $,
but in general there are also other extreme points. The extreme points of the convex set $ \eos $ are
the characteristic functions $ \chi_B $, $ B \in \Sigma $, these are the {\it sharp} classical effects
(in the terminology of classical probability theory, the {\it events}), the other effects are {\it unsharp}
or {\it fuzzy}.---Finally, we recall that $ \sigtop{\fos,\mos} $ is the coarsest topology on $ \fos $ in which
the elements of $ \mos $, considered as linear functionals on $ \fos $, are continuous; this topology is
the restriction of the weak-* Banach-space topology of $ (\mos)' $ to $ \fos $.

Given a bounded linear map $ T \! : \tsh \to \mos $, its dual $ T' \! : \fos \to \bsh $ w.r.t.\ the dual pairs
$ \dual{\tsh}{\bsh} $ and $ \dual{\mos}{\fos} $ is defined according to
\begin{equation}
\langle TV,f \rangle = \langle V,T'f \rangle
\end{equation}
where $ V \in \tsh $ and $ f \in \fos $ are arbitrary; explicitly, Eq.\ (1) reads
\begin{equation}
\int f \, \mathrm{d}(TV) = \mathrm{tr} \, V(T'f).
\end{equation}
The uniquely determined map $ T' $ is linear and bounded and is just the restriction of the Banach-space adjoint map
$ T^* \! : (\mos)' \to \bsh $ to $ \fos $. Throughout this paper, we will understand $T'$ in this sense.

In the next section we review the concept of \emph{statistically (informationally) complete} observables as well as
its relation to the injective affine mappings from the quantum states into the probability measures on some
measurable space $ \os $ or, equivalently, to the injective state-preserving linear maps $ T \! : \tsh \to \mos $,
here called \emph{semiclassical representations of quantum mechanics}. Moreover, it is shown that such a description
of the quantum states by probability measures cannot be extended to a \emph{classical representation
$ (T,S) $ of quantum mechanics}, i.e., $T$ cannot be supplemented by a description $S$ of the quantum mechanical
effects by classical effects such that the quantum probabilities coincide with the corresponding classical ones.

In Section 3 we first observe that a semiclassical representation $T$ can be supplemented by a description of
the quantum effects by functions $ f \in \fos $ so that the quantum probabilities can approximately be written
in terms of classical expressions. This is a consequence of the fact that the range of the dual map $T'$ is
$ \sigma $-dense (i.e., $ \sigtop{\bsh,\tsh} $-dense) in $ \bsh $ or, correspondingly, that the linear hull
of the range of a statistically complete observable is $ \sigma $-dense in $ \bsh $. However, the convex hull
of the range of a statistically complete observable is never $ \sigma $-dense in the set of all quantum effects;
more generally, a set of pairwise coexistent effects cannot be $ \sigma $-dense in $ \eh $. This result and
the following implication belong to our main results. The implication is that the approximate representatives
$ f \in \fos $ of the quantum effects do in general not belong to the set $ \eos $ of all classical effects;
in other words, \emph{approximate classical representations of quantum mechanics} do also not exist.

Hence, a classical description $ (T,S) $ of quantum mechanics that preserves the probabilities cannot relate both
quantum states and effects to their classical counterparts, neither strictly nor approximately. A weakened concept
associates quantum states and effects with arbitrary bounded signed measures and arbitrary real-valued measurable
functions on some measurable space, but preserves the probabilities. This very general concept of the
\emph{quasi-probability representations of quantum mechanics} was recently introduced by Ferrie, Morris, and Emerson
\cite{fer10}. In Section 4 we present the general concept and investigate it in some detail. As an example,
we consider the Wigner phase-space quasi-probability distribution functions together with the inverse
Weyl correspondence which associates self-adjoint Hilbert-Schmidt operators with real functions on phase space;
it is shown that this example gives rise to a quasi-probability representation in a restricted sense. As
a second example, we consider quasi-probability representations of finite-dimensional quantum mechanics.

\section{Semiclassical Representations and\\
Statistically Complete Observables}

In operational quantum mechanics, an observable $F$ is understood to be a normalized effect-valued
(positive-operator-valued) measure (POVM) on some measurable space $ (\Omega,\Sigma) $, i.e., an observable on
$ (\Omega,\Sigma) $ is a map $ F \! : \Sigma \to \eh $ satisfying (i) $ F(\emptyset) = 0 $, $ F(\Omega) = I $
and (ii) $ F(\bigcup_{i=1}^{\infty} B_i) = \sum_{i=1}^{\infty} F(B_i) $ where the sets $ B_i \in \Sigma $ are
mutually disjoint and the sum converges $ \sigma $-weakly (equivalently, weakly or strongly). Given $ W \in \sh $,
$F$ defines a probability measure $ P^F_W $ on $ \os $ according to $ P^F_W(B) := \mathrm{tr} \, WF(B) $, $ P^F_W $
is the probability distribution of $F$ in the state $W$. An observable is called \emph{statistically (informationally)
complete} if the states $ W \in \sh $ are determined by their probability distributions w.r.t.\ $F$, i.e.,
if for any two states $ W_1,W_2 \in \sh $, $ \mathrm{tr} \, W_1F(B) = \mathrm{tr} \, W_2F(B) $ for all
$ B \in \Sigma $ implies $ W_1 = W_2 $. Although it is well known that statistically complete observables
do exist \cite{lud70} and concrete examples can be given \cite{ali77a;77b,bus89,pru84}, we present an abstract proof
of this remarkable fact \cite{sin92,stu97}. To that end and for later purposes, we need the following two lemmata.

\begin{lemma}
Let $F$ be an observable on $ (\Omega,\Sigma) $. Then the following statements are equivalent:
\begin{enumerate}
\item[(i)] $F$ is statistically complete
\item[(ii)] $F$ separates the self-adjoint trace-class operators, i.e., if, for any $ V_1,V_2 \in \tsh $,
\begin{equation}
\mathrm{tr} \, V_1F(B) = \mathrm{tr} \, V_2F(B)
\end{equation}
holds for all $ B \in \Sigma $, then $ V_1 = V_2 $
\item[(iii)] the linear hull of $ F(\Sigma) = \{ F(B) \, | \, B \in \Sigma \} $ is $ \sigma $-dense in $ \bsh $.
\end{enumerate}
\end{lemma}

\begin{proof} Suppose $F$ is statistically complete. Let $ V_1 $ and $ V_2 $ arbitrary positive trace-class operators
and assume that Eq.\ (3) is satisfied for all $ B \in \Sigma $. Setting $ B = \Omega $, it follows that
$ \mathrm{tr} \, V_1 = \mathrm{tr} \, V_2 =: \gamma $. If $ \gamma = 0 $, we obtain $ V_1 = V_2 = 0 $. For
$ \gamma \neq 0 $, divide (3) by $ \gamma $ and notice that $ \frac{1}{\gamma} V_1 = \frac{1}{\gamma} V_2 $
are density operators. Then $ V_1 = V_2 $ is implied. Now let $ V_1,V_2 \in \tsh $ be arbitrary and assume again
the validity of (3) for all $ B \in \Sigma $. Decomposing $ V_1 $ and $ V_2 $ into positive operators, we obtain
\[
\mathrm{tr} \, (V_1^+ - V_1^-)F(B) = \mathrm{tr} \, (V_2^+ - V_2^-)F(B)
\]
or, equivalently,
\[
\mathrm{tr} \, (V_1^+ + V_2^-)F(B) = \mathrm{tr} \, (V_2^+ + V_1^-)F(B).
\]
In consequence, $ V_1^+ + V_2^- = V_2^+ + V_1^- $, that is, $ V_1 = V_2 $.---The implication
(ii) $ \Rightarrow $ (i) is trivial.

If the linear hull of $ F(\Sigma) $ is not $ \sigma $-dense in $ \bsh $, then, according to a well-known consequence
of the Hahn-Banach theorem, there exists a $ \sigma $-continuous linear functional $ \Lambda \neq 0 $ on $ \bsh $
such that $ \Lambda(A) = 0 $ for all $ A \in \overline{\mathrm{lin} \, F(\Sigma)}^{\sigma} $. Since the
$ \sigma $-continuous linear functionals on $ \bsh $ are just those ones that are represented by the elements of
$ \tsh $,
\[
\Lambda(A) = \mathrm{tr} \, VA = 0
\]
holds for some $ V \in \tsh $, $ V \neq 0 $, and all $ A \in \overline{\mathrm{lin} \, F(\Sigma)}^{\sigma} $. Hence,
$ F(\Sigma) $ does not separate $ \tsh $. Thus, statement (ii) implies (iii).

Finally, suppose $ \overline{\mathrm{lin} \, F(\Sigma)}^{\sigma} = \bsh $. Assume that Eq.\ (3) is satisfied for any
$ V_1,V_2 \in \tsh $ and all $ B \in \Sigma $. Considering $ V_1 $ and $ V_2 $ as $ \sigma $-continuous linear
functionals on $ \bsh $, it follows from (3) that
\[
\mathrm{tr} \, V_1A = \mathrm{tr} \, V_2A
\]
for all $ A \in \overline{\mathrm{lin} \, F(\Sigma)}^{\sigma} = \bsh $. Thus, $ V_1 = V_2 $, and $F$ separates
$ \tsh $ and in particular $ \sh $.
\end{proof}

Clearly, statement (ii) of the lemma entails that $F$ separates all (not necessarily self-adjoint) trace-class
operators.---The next lemma is a known very general statement; for reasons of completeness, we give a proof.

\begin{lemma}
Let $ \mathcal{V} $ be a real or complex separable normed space, $ \mathcal{V}' $ its dual, and
$B'$ the closed unit ball of $ \mathcal{V}' $. Then the topology on $B'$ that is induced by the
weak-$*$ topology $ \sigma(\mathcal{V}',\mathcal{V}) $ is second countable.
\end{lemma}
\begin{proof} By definition, $ \sigma(\mathcal{V}',\mathcal{V}) $ is the coarsest topology such that
the linear functionals $ \ell \mapsto \ell(v) = \langle v,\ell \rangle =:(iv)(\ell) $, $ v \in \mathcal{V} $,
$ \ell \in \mathcal{V}' $, are continuous; $ i \! : \mathcal{V} \to \mathcal{V}'' $ is the canonical embedding of
$ \mathcal{V} $ into its bidual. Consequently, the sets
\begin{equation}
U := (iv)^{-1}(O) \cap B' = \{ \ell \in B' \, | \, \ell(v) \in O \}
\end{equation}
where $ v \in \mathcal{V} $ and $O$ is an open set of $ \mathbb{R} $ or $ \mathbb{C} $, respectively, are open w.r.t.
$ \sigma(\mathcal{V}',\mathcal{V}) \cap B' $, and the finite intersections of such sets $U$ constitute a base of
$ \sigma(\mathcal{V}',\mathcal{V}) \cap B' $. Now let $ \{ v_k \}_{k \in \mathbb{N}} $ be a sequence of vectors
being norm-dense in $ \mathcal{V} $, $ \{ q_l \}_{l \in \mathbb{N}} $ a sequence of numbers being dense in
$ \mathbb{R} $ or $ \mathbb{C} $, respectively, and $ m \in \mathbb{N} $. The countably many sets
\begin{eqnarray}
U_{klm} : & = & (iv_k)^{-1}(K_{\frac{1}{m}}(q_l)) \cap B'
            =   \{ \ell \in B' \, | \, \ell(v_k) \in K_{\frac{1}{m}}(q_l) \}                      \nonumber\\
          & = & \left\{ \ell \in B' \left| \, |\ell(v_k) - q_l| < \frac{1}{m} \right. \right\}    \nonumber\\
\end{eqnarray}
where $ K_{\frac{1}{m}}(q_l) $ is the set of all real or complex numbers
$ \xi $ satisfying $ |\xi - q_l| < \frac{1}{m} $, are particular instances of the sets
$U$ according to (4); we show that even the finite intersections of the $ U_{klm} $ constitute a base of
$ \sigma(\mathcal{V}',\mathcal{V}) \cap B' $.

To that end, we first prove that, for $U$ according to (4),
\begin{equation}
U = \bigcup_{U_{klm} \subseteq U} U_{klm}.
\end{equation}
Let $ \ell \in U $. Then $ \ell(v) \in O $, and there exists an $ \varepsilon > 0 $ such that
$ K_{\varepsilon}(\ell(v)) \subseteq O $ where
$ K_{\varepsilon}(\ell(v)) = \{ \xi \, | \, |\xi - \ell(v)| < \varepsilon \} $,
$ \xi \in \mathbb{R} $ or $ \xi \in \mathbb{C} $, respectively. Choose $ m_0 \in \mathbb{N} $ such that
$ \frac{1}{m_0} < \frac{\varepsilon}{2} $, and choose a member $ q_{l_0} $ of the sequence
$ \{ q_l \}_{l \in \mathbb{N}} $ and a member $ v_{k_0} $ of the sequence
$ \{ v_k \}_{k \in \mathbb{N}} $ such that $ |\ell(v) - q_{l_0}| < \frac{1}{2m_0} $ and
$ \|v_{k_0} - v\| < \frac{1}{2m_0} $. It follows that
\[
|\ell(v_{k_0}) - q_{l_0}| \leq |\ell(v_{k_0}) - \ell(v)| + |\ell(v) - q_{l_0}| < \frac{1}{2m_0} + \frac{1}{2m_0}
                                                                               = \frac{1}{m_0}
\]
where $ \|\ell\| \leq 1 $ has been taken into account. Hence, $ \ell \in U_{k_0l_0m_0} $. We further
have to show that $ U_{k_0l_0m_0} \subseteq U $. Therefore, let $ \tilde{\ell} \in U_{k_0l_0m_0} $. Then, from
\begin{eqnarray*}
|\tilde{\ell}(v) - \ell(v)| & \leq &   |\tilde{\ell}(v) - \tilde{\ell}(v_{k_0})|
                                     + |\tilde{\ell}(v_{k_0}) - q_{l_0}| + |q_{l_0} - \ell(v)|     \\
                            &   <  &    \frac{1}{2m_0} + \frac{1}{m_0} + \frac{1}{2m_0}            \\
                            &   =  &    \frac{2}{m_0}                                              \\
                            &   <  &    \varepsilon
\end{eqnarray*}
where $ \|\tilde{\ell}\| \leq 1 $ has been used as well as the definition (5) of the $ U_{klm} $,
it follows that $ \tilde{\ell}(v) \in K_{\varepsilon}(\ell(v)) \subseteq O $, i.e.,
$ \tilde{\ell} \in U $. Hence, $ U_{k_0l_0m_0} \subseteq U $.

Summarizing, we have shown that, for $ \ell \in U $, $ \ell \in U_{k_0l_0m_0} \subseteq U $. Hence,
$ U \subseteq \bigcup_{U_{klm} \subseteq U} U_{klm} \subseteq U $, and assertion (6) has been proved. As a consequence
of (6), the finite intersections of sets $U$ according to (4) are unions of intersections of finitely many sets
$ U_{klm} $. Since the former intersections constitute a base of $ \sigma(\mathcal{V}',\mathcal{V}) \cap B' $,
the finite intersections of the sets $ U_{klm} $ constitute a countable base of
$ \sigma(\mathcal{V}',\mathcal{V}) \cap B' $.
\end{proof}

Using diagonal sequences, one can prove that $B'$ is $ \sigma(\mathcal{V}',\mathcal{V}) $-sequentially compact. Then
one can conclude from Lemma 2 that $B'$ is $ \sigma(\mathcal{V}',\mathcal{V}) $-compact. Thus, for a separable
normed space $ \mathcal{V} $, the Banach-Alaoglu theorem can be proved without using an argument related to
Zorn's lemma.---The weak-* compactness of $B'$ implies its weak-* closeness, the latter can directly be concluded
from the definition of the weak-* topology. In what follows we shall essentially use the second countability
and the closeness of $B'$ in the weak-* topology, but not the compactness.

\begin{theorem}
There exists a statistically complete observable.
\end{theorem}

\begin{proof}
The separability of the Hilbert space $ \hi $ entails the norm separability of the Banach space $ \tsh $,
Lemma 2 then implies that the closed unit ball of $ \bsh $, $ \{ A \in \bsh \, | \, \|A\| \leq 1 \} = [-I,I] $,
is second countable in the topology $ \sigma \cap [-I,I] $. In consequence, the set $ \eh = [0,I] $ is
second countable in the topology $ \sigma \cap [0,I] $. From the second countability we obtain that
$ [-I,I] $ and $ [0,I] $ are separable in the respective topologies, i.e., there exist countable
dense subsets.

Now let $ \{ \widetilde{A}_n \}_{n \in \mathbb{N}} $, $ \widetilde{A}_n \in \eh $, be a sequence $ \sigma $-dense in
$ \eh $ and define a further sequence by
\begin{eqnarray*}
A_1 & := & I - \sum_{i=1}^{\infty} \frac{1}{2^i}\widetilde{A}_i,                    \\
A_n & := & \frac{1}{2^{n-1}}\widetilde{A}_{n-1} \quad \mathrm{for} \quad n\geq 2.
\end{eqnarray*}
Observe that (i) the infinite sum is even norm-convergent,
(ii) $ A_n \in \eh $ for all $ n \in \mathbb{N} $, (iii) $ \sum_{n=1}^{\infty} A_n = I $, and
(iv) $ \overline{\mathrm{lin} \, \{ A_n \, | \, n \in \mathbb{N} \}}^{\sigma} = \bsh $. Finally,
define an observable $F$ on the power set of $ \mathbb{N} $ by
\[
F(B) := \sum_{i \in B} A_i
\]
where $ B \subseteq \mathbb{N} $. Because of (iv) and Lemma 1, $F$ is statistically complete.
\end{proof}

Note that all of Lemma 1 what we have used in the proof of Theorem 1 is the implication (iii) $ \Rightarrow $ (i);
in particular, the implication (ii) $ \Rightarrow $ (iii) which is based on the Hahn-Banach theorem
has not been used. Thus, the proof of Theorem 1 does not involve any argument related to Zorn's lemma,
it is purely constructive. The implication (ii) $ \Rightarrow $ (iii) is, however, essential
for the proof of Theorem 4.

We mention some properties of statistically complete observables \cite{bus95,bus89}. Let $F$ be an observable on
$ \os $. If there is one effect $ F(B) \in F(\Sigma) $ such that $ F(B) \neq \alpha I $, $ 0 \leq \alpha \leq 1 $,
and $ F(B)F(C) = F(C)F(B) $ for all $ C \in \Sigma $, then $F$ cannot be statistically complete. If the range of
$F$ contains a nontrivial projection, $ 0,I \neq F(B) = (F(B))^2 $, then $ F(B) $ commutes with all
the other effects of $ F(\Sigma) $, and the observable cannot be statistically complete. Thus,
a statistically complete observable cannot contain a nontrivial projection and can in particular not be
a nontrivial projection-valued measure. Furthermore, no effect of a statistically complete observable
can have both $0$ and $1$ as eigenvalues.

Now we are ready to introduce the central concept of this paper. A \emph{semiclassical representation
of quantum mechanics} is an affine mapping that assigns to every quantum state
$ W \in \sh $ injectively a probability measure $ \mu \in \sos $ where $ \os $ is some fixed
measurable space. It is well known and not hard to prove that any affine mapping
$ \widetilde{T} $ from $ \sh $ to $ \sos $ can uniquely be extended to a positive linear map
$T$ from $ \tsh $ to $ \mos $. Moreover, if $ \widetilde{T} $ is injective, $T$ is injective.

\begin{defin}
A linear map $ T \! : \tsh \to \mos $ satisfying $ T\sh \subseteq \sos $ is called a \emph{statistical map}. We call
an injective statistical map $ T \! : \tsh \to \mos $ a \emph{semiclassical representation of quantum mechanics on
$ \os $}.

If $ (\Omega,\Sigma,\lambda) $ is a $ \sigma $-finite measure space, $ L_{\mathbb{R}}^{1}(\Omega,\Sigma,\lambda) $
the corresponding space of real-valued $ L^1 $-functions, and $ \mathcal{S}(\Omega,\Sigma,\lambda) $
the convex set of all normalized $ L^1 $-functions $ \varrho \geq 0 $, then an injective linear map
$ \widehat{T} \! : \tsh \to L_{\mathbb{R}}^{1}(\Omega,\Sigma,\lambda) $ satisfying
$ \widehat{T}\sh \subseteq \mathcal{S}(\Omega,\Sigma,\lambda) $ is called
a \emph{semiclassical representation on $ (\Omega,\Sigma,\lambda) $}.
\end{defin}

The functions $ \varrho \in \mathcal{S}(\Omega,\Sigma,\lambda) $ are probability densities;
$ \mathcal{S}(\Omega,\Sigma,\lambda) $ can be identified with that subset of $ \sos $ that consists of all
$ \lambda $-absolutely continuous probability measures. In particular situations (for instance, if
$ \Omega $ is the phase space equipped with the Lebesgue measure on its Borel sets),
$ \mathcal{S}(\Omega,\Sigma,\lambda) $ is a suitable model for classical statistical states. In this case
the classical effects are described by the convex set $ \mathcal{E}(\Omega,\Sigma,\lambda) $ of all
$ L^{\infty} $-functions $f$ satisfying \linebreak $ 0 \leq f \leq \chi_{\Omega} $ $ \lambda $-almost everywhere,
and the probability for the occurrence of the effect $ f \in \mathcal{S}(\Omega,\Sigma,\lambda) $ in the state
$ \varrho \in \mathcal{S}(\Omega,\Sigma,\lambda) $ is $ \int \varrho f \, \mathrm{d}\lambda $. An advantage
of the classical statistical model based on a $ \sigma $-finite measure space (and being used in particular
for usual classical statistical mechanics) is that $ L_{\mathbb{R}}^{\infty}(\Omega,\Sigma,\lambda) $ is
the dual space of $ L_{\mathbb{R}}^1(\Omega,\Sigma,\lambda) $.---In \cite{bus93,hel93,stu94;98,stu97}
a semiclassical representation of quantum mechanics was called a ``classical representation of quantum mechanics''.

Some simple properties of semiclassical representations are stated in the next lemma. In this context,
recall that for a bounded linear map $ T \! : \tsh \to \mos $ a unique adjoint
$ T' \! : \fos \to \bsh $ is defined by Eqs.\ (1) and (2). In the case of a bounded linear map
$ \widehat{T} \! : \tsh \to L_{\mathbb{R}}^{1}(\Omega,\Sigma,\lambda) $,
$ \widehat{T}' = \widehat{T}^* $ is the usual Banach-space adjoint map
$ \widehat{T}^* \! : L_{\mathbb{R}}^{\infty}(\Omega,\Sigma,\lambda) \to \bsh $.

\begin{lemma}
A statistical map $T$ is positive and bounded with $ \|T\| = 1 $. The property
$ T\sh \subseteq \sos $ of a bounded linear map $ T \! : \tsh \to \mos $ is equivalent to
$ T' \geq 0 $ and $ T'\chi_{\Omega} = I $. The latter two conditions imply $ T'\eos \subseteq \eh $ as well as
$ \|T'\| = 1 $. For linear maps $ \widehat{T} \! : \tsh \to L_{\mathbb{R}}^{1}(\Omega,\Sigma,\lambda) $,
the analogous statements hold.
\end{lemma}

\begin{proof}
Let $ T \! : \tsh \to \mos $ be a statistical map. Then $T$ is positive, and for $ W \in \sh $ we have that
$ \|TW\| = 1 = \|W\|_{\mathrm{tr}} $, the subscript $ \mathrm{tr} $ indicating the trace norm. In consequence,
$ \|TV\| = \|V\|_{\mathrm{tr}} $ for all $ V \geq 0 $. Writing $ V = V^+ - V^- $ where $ V^+ $ and $ V^- $
are the positive and the negative part of an arbitrary trace-class operator $V$, we obtain $ \|TV\| \leq
\|TV^+\| + \|TV^-\| = \|V^+\|_{\mathrm{tr}} + \|V^-\|_{\mathrm{tr}} = \mathrm{tr} \, |V| = \|V\|_{\mathrm{tr}} $
which implies $ \|T\| \leq 1 $. Since $ \|TW\| = \|W\|_{\mathrm{tr}} $ for $ W \in \sh $, we conclude that
$ \|T\| = 1 $.

The map $T'$ is also positive; from $ T\sh \subseteq \sos $ it follows for $ W \in \sh $ that
$ 1 = (TW)(\Omega)= \int \chi_{\Omega} \, \mathrm{d}(TW) = \mathrm{tr} \, W(T'\chi_{\Omega}) $, and
$ \mathrm{tr} \, W(T'\chi_{\Omega}) = 1 $ for all $ W \in \sh $ implies $ T'\chi_{\Omega} = I $. Conversely, from
$ T' \geq 0 $ and $ T'\chi_{\Omega} = I $ it follows that $ T \geq 0 $ and, for $ W \in \sh $,
$ 1 = \mathrm{tr} \, W(T'\chi_{\Omega}) = (TW)(\Omega) $; therefore, $ TW \geq 0 $ and $ (TW)(\Omega) = 1 $,
that is, $ TW \in \sos $. The implication $ T'\eos \subseteq \eh $ is clear because
$ \eos = [0,\chi_{\Omega}] $ and $ \eh = [0,I] $. From the further consequence that
$ T'[-\chi_{\Omega},\chi_{\Omega}] \subseteq [-I,I] $ and the fact that
$ [-\chi_{\Omega},\chi_{\Omega}] $ and $ [-I,I] $ are the closed unit balls of $ \fos $ and $ \bsh $,
respectively, we obtain $ \|T'\| \leq 1 $ and finally, since $ T'\chi_{\Omega} = I $, $ \|T'\| = 1 $.
\end{proof}

The statistical maps are in one-one correspondence with the observables, as the following theorem shows.

\begin{theorem}
Every observable $F$ on $ \os $ defines a statistical map $ T \! : \tsh \to \mos $ by
\begin{equation}
(TV)(B) := \mathrm{tr} \, VF(B)
\end{equation}
where $ V \in \tsh $ and $ B \in \Sigma $. Conversely, for every statistical map
$ T \! : \tsh \to \mos $ there exists a uniquely determined observable $ F \! : \Sigma \to \eh $ such that
$ TV = \mathrm{tr} \, VF(\, . \,) $. In particular, for $ W \in \sh $, $TW$ is just
the probability distribution of $F$, briefly, $ TW = P^F_W $. Moreover, $T'$ and $F$ are related by
\[
T'f = \int f \, \mathrm{d}F
\]
or, equivalently, by
\begin{equation}
T'\chi_{B} = F(B)
\end{equation}
where $ f \in \fos $, $ B \in \Sigma $, and the integral is understood in the $ \sigma $-weak sense.

In particular, Eq.\ (7) establishes a one-one correspondence between the semiclassical representations on
$ \os $ and the statistically complete observables on $ \os $.
\end{theorem}

\begin{proof}
Given an observable $ F \! : \Sigma \to \eh $, it is obvious that a statistical map
$ T \! : \tsh \to \mos $ is defined by (7). If $F$ is statistically complete, then, according to Lemma 1, its range
$ F(\Sigma) $ separates the elements of $ \tsh $, i.e., $T$ is injective.

Now assume $T$ is a statistical map. Define $ F(B) := T'\chi_{B} $. Then Lemma 3 implies
$ F(B) \in \eh $ and $ F(\Omega) = T'\chi_{\Omega} = I $, and from
$ \mathrm{tr} \, VF(B) = \mathrm{tr} \, V(T'\chi_{B}) $ $ = \int \chi_{B} \, \mathrm{d}(TV) = (TV)(B) $
it follows that $ TV = \mathrm{tr} \, VF(\, . \,) $. Next we show that $F$ is $ \sigma $-additive. Taking
a sequence of disjoint sets $ B_i \in \Sigma $, we obtain
\begin{eqnarray*}
\mathrm{tr} \, VF \left( \bigcup_{i=1}^{\infty} B_i \right)
& = & (TV) \left( \bigcup_{i=1}^{\infty} B_i \right)
  =   \sum_{i=1}^{\infty} (TV)(B_i)                                           \\
& = & \lim_{n \to \infty} \sum_{i=1}^{n} (TV)(B_i)                            \\
& = & \lim_{n \to \infty} \mathrm{tr} \left( V\sum_{i=1}^{n} F(B_i) \right)
\end{eqnarray*}
for all $ V \in \tsh $. Consequently,
\[
F\left( \bigcup_{i=1}^{\infty} B_i \right)
= \sigma\mbox{-} \! \! \lim_{n \rightarrow \infty} \sum_{i=1}^{n} F(B_i)
= \sum_{i=1}^{\infty} F(B_i).
\]
Hence, $F$ is an observable satisfying $ TV = \mathrm{tr} \, VF(\, . \,) $. It is obvious that such an observable is
uniquely determined. Moreover, if $T$ is injective, $F$ is statistically complete.

Finally, for $ f \in \fos $ we have
\[
\mathrm{tr} \, V(T'f) = \int f \, \mathrm{d}(TV) = \int f \, \mathrm{d}({\rm tr} \, VF(\, . \,))
                              = \mathrm{tr} \left( V \int f \, \mathrm{d}F \right).
\]
Thus, $ T'f = \int f \, \mathrm{d}F $.
\end{proof}

If the measurable space $ \os $ satisfies the conditions that (i) $ \Omega $ contains the one-point sets
$ \{ \omega \} $, $ \omega \in \Omega $, and (ii) the extreme points of the convex set $ \sos $ are just
the Dirac mesures $ \delta_{\omega} $, then it is easy to see that, if $ \dim \hi \geq 2 $,
a semiclassical representation cannot map $ \sh $ onto $ \sos $ (the condition on $ \os $ is satisfied
for a second countable Hausdorff space $ \Omega $ with $ \Sigma $ being the $ \sigma $-algebra
of its Borel sets). Namely, let $ \fii_1 $ and $ \fii_2 $ be two orthogonal unit vectors of $ \hi $,
and let $ \chi_1 $ and $ \chi_2 $ be a second such pair which differs from the first by more than phase factors
but generates the same two-dimensional subspace of $ \hi $. Then
\begin{equation}
W := \frac{1}{2}(P_{\fii_1} + P_{\fii_2}) = \frac{1}{2} (P_{\chi_1} + P_{\chi_2}) \in \sh
\end{equation}
where $ P_{\fii_1} = \kb{\fii_1}{\fii_1} $ etc., and
\begin{equation}
\mu := TW = \widetilde{T}(W) = \frac{1}{2}\widetilde{T}(P_{\fii_1}) + \frac{1}{2}\widetilde{T}(P_{\fii_2})
                             = \frac{1}{2}\widetilde{T}(P_{\chi_1}) + \frac{1}{2}\widetilde{T}(P_{\chi_2})
\end{equation}
where $ \widetilde{T} \! : \sh \to \sos $ is the affine restriction of a semiclassical representation $T$. If
$ \widetilde{T} $ were bijective, then $ \widetilde{T} $ would be an affine isomorphism and would map
the extreme points $ P_{\fii} = \kb{\fii}{\fii} $ of $ \sh $ onto the extreme points
$ \delta_{\omega} $ of $ \sos $. Hence, Eq.\ (10) entails that
\[
\mu = \frac{1}{2}\delta_{\omega_1}  + \frac{1}{2}\delta_{\omega_2}
    = \frac{1}{2}\delta_{\omega'_1} + \frac{1}{2}\delta_{\omega'_2}
\]
where the four points $ \omega_1,\omega_2,\omega'_1,\omega'_2 \in \Omega $ are different. This is
a contradiction. Thus, $ \widetilde{T} $ cannot be bijective and $T$ cannot fulfil
$ T\sh = \sos $.---The existence of different convex decompositions of the state
$W$ according to (9) and the uniqueness of the convex linear combination
$ \mu = \frac{1}{2}\delta_{\omega_1}  + \frac{1}{2}\delta_{\omega_2} $ expresses
one of the crucial differences between quantum and classical probability.

We consider another argument that proves the same result under the only condition $ \dim \hi \geq 2 $. If
a semiclassical representation $T$ mapped $ \sh $ onto $ \sos $, it would map the positive cone of $ \tsh $
bijectively onto the positive cone of $ \mos $. In consequence, $T$ would be a bijective positive linear map
with positive inverse $ T^{-1} $; in particular, $T$ would be an order isomorphism between the ordered Banach spaces
$ \tsh $ and $ \mos $. This is a contradiction since $ \mos $ is a vector lattice whereas $ \tsh $ is not. Hence,
$ T\sh $ must be a proper subset of $ \sos $, and $T$ cannot be a bijective positive linear map with positive inverse
$ T^{-1} $. However, a semiclassical representation can be bijective (see Section 4); in this case $T$ is
a bijective positive linear map, but $ T^{-1} $ is not positive.

A third interesting argumentation proceeds as follows. If $ \dim \hi \geq 2 $, then
$ \dim \tsh \geq 4 $ as well as $ \dim \mos \geq 4 $, the latter because $T$ is injective. Consequently,
there exist, without any further assumption on $ (\Omega,\Sigma) $, two different Dirac measures
$ \delta_{\omega_1} $ and $ \delta_{\omega_2} $. Now, if $ T\sh = \sos $ were satisfied, there would exist
$ W_1,W_2 \in \sh $ such that $ TW_1 = \delta_{\omega_1} $ and $ TW_2 = \delta_{\omega_2} $ or, equivalently,
\begin{eqnarray*}
\mathrm{tr} \, W_1F(\, . \,) & = & \delta_{\omega_1}   \\
\mathrm{tr} \, W_2F(\, . \,) & = & \delta_{\omega_2}
\end{eqnarray*}
where $F$ is the statistically complete observable corresponding to $T$. For a set
$ B \in \Sigma $ containing $ \omega_1 $ but not $ \omega_2 $, it follows that
$ \mathrm{tr} \, W_1F(B) = 1 $ and $ \mathrm{tr} \, W_2F(B) = 0 $. Therefore, the effect
$ F(B) $ would have the eigenvalues $0$ and $1$, which is, as already mentioned, not possible
for a statistically complete observable.

In view of the fact that a statistical physical theory is based on states, effects, and their respective probabilities, we supplement the definition of a semiclassical representation by the representation of the quantum mechanical effects
by functions. A \emph{classical representation of quantum mechanics} \cite{fer10} is a pair
$ (\widetilde{T},\widetilde{S}) $ of affine mappings
$ \widetilde{T} \! : \sh \to \sos $ and $ \widetilde{S} \! : \eh \to \eos $ such that
$ \widetilde{S}0 = 0 $ and
\begin{equation}
\mathrm{tr} \, WA = \int f \, \mathrm{d}\mu
\end{equation}
where $ W \in \sh $, $ A \in \eh $, $ \mu = \widetilde{T}W $, and $ f = \widetilde{S}A $ (similar definitions
are given in \cite{fer11,fer08;09}). From Eq.\ (11) it follows that
$ \widetilde{T} $ and $ \widetilde{S} $ are injective (see Lemma~4 below). The condition
$ \widetilde{S}0 = 0 $ is necessary to extend $ \widetilde{S} $, like $ \widetilde{T} $,
to a linear map and does not follow from (11) since, for $ \dim \hi \geq 2 $, $ \widetilde{T} $ is not surjective;
the linear extension $S$ of $ \widetilde{S} $ maps $ \bsh $ into $ \fos $, $S$ is uniquely determined and positive.

\begin{defin}
We call a pair $ (T,S) $ of linear maps $ T \! : \tsh \to \mos $ and $ S \! : \bsh \to \fos $ a \emph{classical
representation of quantum mechanics on $ \os $} if
\begin{enumerate}
\item[(i)] $ T\sh \subseteq \sos $
\item[(ii)] $ S\eh \subseteq \eos $
\item[(iii)] for all $ W \in \sh $ and all $ A \in \eh $,
\begin{equation}
\mathrm{tr} \, WA = \int SA \, \mathrm{d}(TW).
\end{equation}
\end{enumerate}
\end{defin}

\begin{lemma}
For a classical representation $ (T,S) $, the maps $T$ and $S$ are injective. In particular,
$T$ is a semiclassical representation; furthermore, $S$ is positive and bounded with
$ \|S\| = \|SI\| = 1 $. The statement $ \mathrm{tr} \, WA = 1 $ for $ W \in \sh $ and $ A \in \eh $ is equivalent to
$ f = \chi_{\Omega} $ $ \mu $-almost everywhere where $ f = SA $ and $ \mu = TW $. Finally,
\begin{equation}
\mathrm{tr} \, VA = \int SA \, \mathrm{d}(TV)
\end{equation}
holds for all $ V \in \tsh $ and all $ A \in \bsh $.
\end{lemma}

\begin{proof}
Eq.\ (13) follows from (12) by linearity. Now, $ TV_1 = TV_2 $ for $ V_1,V_2 \in \tsh $ implies that
\[
\mathrm{tr} \, V_1A = \int SA \, \mathrm{d}(TV_1) = \int SA \, \mathrm{d}(TV_2) = \mathrm{tr} \, V_2A
\]
for all $ A \in \bsh $, consequently, $ V_1 = V_2 $. Hence, $T$ is injective and thus a semiclassical
representation. Analogously, $S$ is injective.

Clearly, $S$ is a positive linear map. Since the closed unit ball of $ \bsh $ is just the interval $ [-I,I] $,
it follows that, for $ A \in \bsh $ with $ \|A\| \leq 1 $,
\[
-SI \leq SA \leq SI
\]
holds or, equivalently, $ |SA| \leq SI $. Therefore,
$ \|SA\| = \sup_{\omega \in \Omega} |(SA)(\omega)| \leq \sup_{\omega \in \Omega} (SI)(\omega) = \|SI\| $,
and in consequence, $ \|S\| = \sup_{\|A\| \leq 1} \|SA\| \leq \|SI\| $. Since $ \|I\| = 1 $,
$ \|S\| \geq \|SI\| $ holds, and we obtain $ \|S\| = \|SI\| $.

Let $ W \in \sh $, $ A \in \eh $ and $ \mu = TW $, $ f = SA $. Then the statement $ \mathrm{tr} \, WA = 1 $
is equivalent to $ \int f \, \mathrm{d}\mu = 1 $, and the latter can equivalently be reformulated as
$ \int (\chi_{\Omega} - f) \, \mathrm{d}\mu = 0 $. But this means $ f = \chi_{\Omega} $ $ \mu $-almost everywhere
because $ \chi_{\Omega} - f \geq 0 $.---For $ A = I $ we obtain $ f = SI = \chi_{\Omega} $ $ \mu $-a.e.\ for
$ \mu = TW $ and all $ W \in \sh $. Thus, in particular, $ \|SI\| = 1 $, and we conclude that $ \|S\| = \|SI\| = 1 $.
\end{proof}

We remark that Definition 2 does not imply that $ SI = \chi_{\Omega} $. This can be seen by
embedding the measurable space $ \os $ into a larger one, $ (\Omega',\Sigma') $, i.e.,
$ \Omega \subset \Omega' \neq \Omega $, $ \Sigma \subset \Sigma' $, and
$ \Sigma' \cap \Omega = \{ B \in \Sigma' \, | \, B \subseteq \Omega \} = \Sigma $. Then
a classical representation $ (T,S) $ on $ \os $ can be understood as such a one on $ (\Omega',\Sigma') $
satisfying $ SI \neq \chi_{\Omega'} $.---In view of Theorem 3 below, this remark as well as Lemma~4
are meaningless; however, some of the statements on classical representations transfer to the quasi-probability
representations introduced in Section 4.

According to Ludwig \cite{lud70,lud83}, two effects $ F,G \in \eh $ are called \emph{coexistent} if there exist
effects $ F',G',H \in \eh $ such that $ F = F' + H $, $ G = G' + H $, and $ F' + G' + H \leq I $. The
underlying interpretation is that coexistent effects can be measured ``simultaneously,'' i.e.,
together by one measuring device. If $F$ and $G$ are orthogonal projections, their coexistence
is equivalent to $ FG = GF $; $F'$, $G'$, and $H$ are then uniquely determined, namely,
$ H = FG = GF $, $ F' = F(I - G) $, $ G' = G(I - F) $. It is easy to see that two effects
are coexistent if and only if they are contained in the range of some observable.---Two classical effects
$ f,g \in \eos $ are always coexistent, that is, there exist functions $ f',g',h \in \eos $ such that
$ f = f' + h $, $ g = g' + h $, and $ f' + g' + h \leq \chi_{\Omega} $. In fact, define
$ h := \inf \{ f,g \} $, $ f' := f - h $, $ g' := g - h $. If $f$ and $g$ are characteristic functions, then
$f'$, $g'$, and $h$ are uniquely determined characteristic functions.

\begin{theorem} \mbox{}
\begin{enumerate}
\item[(a)] A linear map $ L \! : \fos \to \bsh $ cannot have the property $ L\eos = \eh $ (provided that
$ \dim \hi \geq 2 $). In particular, for a statistical map $T$ or a semiclassical representation,
$ T'\eos $ must be a proper subset of $ \eh $.
\item[(b)] A classical representation of quantum mechanics ($ \dim \hi \geq 2 $) does not exist.
\end{enumerate}
\end{theorem}

\begin{proof}
Assume a linear map $ L \! : \fos \to \bsh $ satisfied $ L\eos = \eh $. Then, for any two effects
$ F,G \in \eh $, there would exist classical effects $ f,g \in \eos $ such that $ F = Lf $ and $ G = Lg $. Since
$f$ and $g$ are coexistent, there exist $ f',g',h \in \eos $ such that $ f = f' + h $, $ g = g' + h $ and
$ f' + g' + h \in \eos $. It follows that
\begin{eqnarray*}
F & = & Lf' + Lh   \\
G & = & Lg' + Lh
\end{eqnarray*}
where $ Lf',Lg',Lh \in \eh $ and $ Lf' + Lg' + Lh = L(f' + g' + h) \in \eh $. Hence, $F$ and $G$ would be
coexistent effects, which is a contradiction ($ \dim \hi \geq 2 $) as they are arbitrary. Thus,
$ L\eos \neq \eh $.---By Lemma 3, $ T' \! : \fos \to \bsh $ satisfies $ T'\eos \subseteq \eh $;
in view of the preceding result, we conclude that $ T'\eos \subset \eh $.

Now assume $ (T,S) $ were a classical representation of quantum mechanics. From Eq.\ (13) we obtain that, for all
$ V \in \tsh $ and all $ A \in \bsh $,
\[
\mathrm{tr} \, VA = \int SA \, \mathrm{d}(TV) = \mathrm{tr} \, V(T'SA).
\]
This implies that, for all $ A \in \bsh $, $ A = T'SA $. In particular, for $ A \in \eh $, $ A = T'f $ holds with
$ f := SA \in \eos $; therefore, $ T'\eos \supseteq \eh $. But this contradicts part (a) of the theorem. Hence,
a classical representation does not exist.
\end{proof}

Some other simple arguments to prove the first statement of part (a) of Theorem 3 can be given
under the additional assumption that the linear map $L$ is injective. If $L$ is injective and
$ L\eos = \eh $ were satisfied, then $L$ would be a bijective positive linear map with positive inverse
$ L^{-1} $ and in particular an order isomorphism between the ordered Banach spaces $ \fos $ and $ \bsh $,
in contradiction to the fact that the former is a vector lattice whereas the latter is not
($ \dim \hi \geq 2 $). Moreover, $L$ would induce an order isomorphism between $ \eos $ and $ \eh $,
but again the former is a lattice whereas the latter is not. Finally, this order isomorphism would also be
an affine isomorphism between $ \eos $ and $ \eh $, in consequence, the extreme boundaries of
$ \eos $ and  $ \eh $ would be order isomorphic. Again, this is a contradiction since the characteristic functions
$ \chi_B $, $ B \in \Sigma $, constitute a Boolean lattice whereas the orthogonal projections of
$ \hi $ form a non-Boolean lattice.

We notice that, for a statistical map $T$ and in particular for a semiclassical representation, any two effects
of the convex set $ T'\eos $ are coexistent. From this the second statement of part (a) of the theorem follows directly.

\section{Nonexistence of Approximate\\
Classical Representations}

By means of a semiclassical representation, the quantum states can be identified with probability measures and
probability densities, respectively. That is, the quantum states can be represented by classical states. Although
the representation of the quantum mechanical effects by classical effects is, according to Theorem 3,
not possible in this context, the former ones can at least approximately be represented by functions
\cite{sin92,stu97}. The precise formulation of this statement is presented by the following theorem;
for the purpose of its proof, recall that $ \sigma(\bsh,\tsh) = \sigma(\bsh,\sh) $ and that consequently
a $ \sigma $-neighborhood base of $ A \in \bsh $ is given by the open sets
\begin{equation}
U(A;W_1,\ldots,W_n;\varepsilon) := \{ \widetilde{A} \in \bsh \, | \,
|\mathrm{tr} \, W_i\widetilde{A} - \mathrm{tr} \, W_iA| < \varepsilon \ \mathrm{for} \ i= 1,\ldots,n \}
\end{equation}
where $ \varepsilon > 0 $ and $ W_i \in \sh $ (cf.\ Eqs.\ (4) and (5)). Also note the interpretation of
$ \sigma(\bsh,\tsh) $ as the topology of the physical approximation of effects which results from (14). An effect
$ A \in \eh $ is physically approximated by $ \widetilde{A} \in \eh $ if in many (but finitely many) states
$ W_1,\ldots,W_n $ the probabilities $ \mathrm{tr} \, W_i\widetilde{A} $ differ from $ \mathrm{tr} \, W_iA $
by an amount less than a small $ \varepsilon > 0 $. This statement can be tested experimentally and can
be characterized mathematically by $ \widetilde{A} \in U(A;W_1,\ldots,W_n;\varepsilon) $.

\begin{theorem}
Let $ T \!\! : \tsh \to \mos $ be a semiclassical representation on the measurable space $ \os $. Then for every
$ A \in \eh $, every $ \varepsilon > 0 $, and any finitely many states $ W_1,\ldots,W_n \in \sh $ there exists
a function $ f \in \fos $ such that
\begin{equation}
\left| {\rm tr} \, W_i A - \int f \, \mathrm{d}\mu_i \right| < \varepsilon
\end{equation}
holds where $ \mu_i := TW_i $ ($ i=1,\ldots,n $).
\end{theorem}

\begin{proof} Since $ T \!\! : \tsh \to \mos $ is an injective linear map and $ \fos $ separates $ \mos $,
the range of $ T' \! : \fos \to \bsh $ is a $ \sigma $-dense subspace of $ \bsh $. This is a consequence
of a general result in duality theory, but in our case it also follows from Lemma 1. Using
the statistically complete observable $F$ that corresponds to $T$ according to Theorem 2, we obtain
\[
R(T') := T'\fos \supseteq {\rm lin} \, \{ T'\chi_B \, | \, B \in \Sigma \} = {\rm lin} \, F(\Sigma)
\]
where we have taken account of Eq.\ (8). By Lemma 1, the linear hull of $ F(\Sigma) $ is
$ \sigma $-dense in $ \bsh $, therefore, $ R(T') $ also.

From the $ \sigma $-denseness of $ R(T') $ in $ \bsh $ and Eq.\ (14) it follows that, for every $ A \in \bsh $,
every $ \varepsilon > 0 $, and any $ W_1,\ldots,W_n \in \sh $, there exists a function $ f \in \fos $ satisfying
\[
| {\rm tr} \, W_i A - {\rm tr} \, W_i (T'f) | < \varepsilon.
\]
Now, the assertion is implied by $ {\rm tr} \, W_i (T'f) = \int f \, \mathrm{d}(TW_i) = \int f \, \mathrm{d}\mu_i $
and the particular choice $ A \in \eh $.
\end{proof}

If $ \widehat{T} \!\! : \tsh \to L^1_{\mathbb{R}}(\Omega,\Sigma,\lambda) $ is a semiclassical representation on the
$ \sigma $-finite measure space $ (\Omega,\Sigma,\lambda) $, then the analog of the theorem reads as follows. For
every $ A \in \eh $, every $ \varepsilon > 0 $, and any $ W_1,\ldots,W_n \in \sh $ there exists a function
$ f \in L^{\infty}_{\mathbb{R}}(\Omega,\Sigma,\lambda) $ such that
\begin{equation}
\left| {\rm tr} \, W_i A - \int \varrho_i f \, \mathrm{d}\lambda \right| < \varepsilon
\end{equation}
holds where $ \varrho_i := \widehat{T} W_i $ ($ i=1,\ldots,n$). This result can be concluded from Theorem 4 or
from duality theory (cf.\ \cite{stu92}).

Theorem 4 as well as statement (16) mean that the quantum mechanical probabilities can, at least in arbitrarily
good physical approximation, be represented by corresponding classical expressions. The approximation involved
in Theorem 4 and statement (16) is physical in the sense that probabilities cannot be measured exactly and in the
laboratory physicists are not able to prepare more than finitely many states. In particular, one can work with the
same small $ \varepsilon > 0 $ and the same many states $ W_1,\ldots,W_n $ for all effects.

If, for every effect $ A \in \eh $ and every desired accuracy, it should be possible to choose a representing function
$f$ as a classical effect, i.e., if every $ \sigma $-neighborhood $ U(A;W_1,\ldots,W_n;\varepsilon) $ of
$ A \in \eh $ should contain an element $ T'f $ with $ f \in \eos $ (and not only with $ f \in \fos $),
then the semiclassical representation $T$ together with approximate classical representatives
$ f \in \eos $ for the effects $ A \in \eh $ can, in the sense of Definition 2, be said to constitute
an \emph{approximate classical representation of quantum mechanics on $ (\Omega,\Sigma) $}. Such an approximate
classical representation, induced by a semiclassical representation $T$, exists if and only if the convex set
$ T'\eos $ should be $ \sigma $-dense in $ \eh $. But this is not possible, as we shall conclude
from the next theorem below.

A set $M$ of quantum mechanical effects is called \emph{coexistent} if $M$ is contained in the range
of some observable; in particular, any two effects of a coexistent set are coexistent in the sense
of the definition given earlier in the context of Theorem 3. However, a set of effects being
\emph{pairwise coexistent} need not be coexistent, as can be seen by the following example. Let
$ \varphi_1,\varphi_2 \in \hi $ be two orthogonal unit vectors, let $ \psi_1 := \varphi_1 $,
$ \psi_2 := \frac{\sqrt{3}}{2} \varphi_1 + \frac{1}{2} \varphi_2 $,
$ \psi_3 := \frac{1}{2} \varphi_1 + \frac{\sqrt{3}}{2} \varphi_2 $, and define the effects
$ G_i := \frac{1}{2}P_{\psi_i} = \frac{1}{2} |\psi_i\rangle \langle\psi_i| $, $ i = 1,2,3 $. These three effects
are pairwise coexistent. Since $ G_1 + G_2 + G_3 \not\leq I $ and any two different conditions of the form
$ A \leq G_i $, $ A \in \eh $, imply that $ A = 0 $, there is no observable $F$ on some measurable space
$ \os $ such that $ G_1,G_2,G_3 \in F(\Sigma) $. Thus, the set $ \{G_1,G_2,G_3\} $ is not coexistent.

\begin{lemma}
Let $ M \subset \eh $ be a set of pairwise coexistent effects. Then any two (finite) convex linear combinations of effects of $M$ are coexistent.
\end{lemma}

\begin{proof} Let $ F_1,\ldots,F_m \in M $ and $ G_1,\ldots,G_n \in M $, and
\[
F := \sum_{i=1}^{m} \alpha_i F_i, \qquad G := \sum_{j=1}^{n} \beta_j G_j
\]
where $ \alpha_i \geq 0 $, $ \sum_{i=1}^{m} \alpha_i = 1 $, $ \beta_j \geq 0 $,
$ \sum_{j=1}^{n} \beta_j = 1 $. Consider the numbers
\[
0,\alpha_1,\alpha_1 + \alpha_2,\ldots,\alpha_1 + \ldots + \alpha_{m-1},1,
      \beta_1,\beta_1 + \beta_2,\ldots,\beta_1 + \ldots + \beta_{n-1},
\]
write them in their natural order and call them $ \gamma_l $, $ l = 0,1,\ldots,m + n -1 $:
\[
0 = \gamma_0 \leq \gamma_1 \leq \ldots \leq \gamma_{m+n-2} \leq \gamma_{m+n-1} = 1.
\]
Furthermore, let
\[
\delta_l := \gamma_l - \gamma_{l-1}, \qquad l \geq 1,
\]
and
\begin{eqnarray*}
I_1 & := & \{ l \, | \, 0 < \gamma_l \leq \alpha_1 \},                                                  \\
I_i & := & \{ l \, | \, \alpha_1 + \ldots + \alpha_{i-1} < \gamma_l \leq \alpha_1 + \ldots \alpha_i \},
                                                                    \qquad  i = 2,\ldots,m.
\end{eqnarray*}
Then
\begin{eqnarray*}
\alpha_1 & = & \sum_{l \in I_1} \delta_l,                                                                       \\
\alpha_i & = & (\alpha_1 + \ldots + \alpha_i) - (\alpha_1 + \ldots + \alpha_{i-1}) = \sum_{l \in I_i} \delta_l,
                                                                                    \qquad  i = 2,\ldots,m
\end{eqnarray*}
(if $ \alpha_i = 0 $ for some $ i = 1,\ldots,m $, then $ I_i = \emptyset $ and, by definition,
$ \sum_{l \in I_i} \delta_l = 0 $), consequently,
\begin{equation}
F = \sum_{i=1}^{m} \alpha_i F_i = \sum_{i=1}^{m} \left( \sum_{l \in I_i} \delta_l \right) F_i
  = \sum_{l=1}^{m+n-1} \delta_l \widetilde{F}_l
\end{equation}
where $ \widetilde{F}_l := F_i $ for $ l \in I_i $.

Analogously, defining
\begin{eqnarray*}
J_1 & := & \{ l \, | \, 0 < \gamma_l \leq \beta_1 \},                                               \\
J_j & := & \{ l \, | \, \beta_1 + \ldots + \beta_{i-1} < \gamma_l \leq \beta_1 + \ldots \beta_i \},
                                                                  \qquad  j = 2,\ldots,n,
\end{eqnarray*}
and $ \widetilde{G}_l := G_j $ for $ l \in J_j $, we obtain
\begin{equation}
G = \sum_{l=1}^{m+n-1} \delta_l \widetilde{G}_l.
\end{equation}
Since the effects $ \widetilde{F}_l $ and $ \widetilde{G}_l $ are coexistent, there are effects
$ A_{l1}, A_{l2}, A_{l0} \in \eh $ such that
\begin{eqnarray}
\widetilde{F}_l & = & A_{l1} + A_{l0}    \nonumber\\
\widetilde{G}_l & = & A_{l2} + A_{l0}    \nonumber\\
\end{eqnarray}
and $ A_{l1} + A_{l2} + A_{l0} \leq I $. From Eqs.\ (17)--(19) it follows that
\begin{eqnarray*}
F & = & \sum_{l=1}^{m+n-1} \delta_l A_{l1} + \sum_{l=1}^{m+n-1} \delta_l A_{l0} =: A_1 + A_0     \\
G & = & \sum_{l=1}^{m+n-1} \delta_l A_{l2} + \sum_{l=1}^{m+n-1} \delta_l A_{l0} =: A_2 + A_0
\end{eqnarray*}
where, since $ \delta_l \geq 0 $ and $ \sum_{l=1}^{m+n-1} \delta_l = 1 $, $ A_1,A_2,A_0 \in \eh $ and
\[
A_1 + A_2 + A_0 = \sum_{l=1}^{m+n-1} \delta_l (A_{l1} + A_{l2} + A_{l0}) \leq I.
\]
Hence, $F$ and $G$ are coexistent.
\end{proof}

The lemma states that the convex hull of a set of pairwise coexistent effects is also pairwise coexistent. Assuming
the full coexistence of the set $M$, i.e., $ M \subseteq F(\Sigma) $ where $F$ is some observable on $ \os $,
an alternative proof of the lemma can be given. Defining a statistical map $ T \! : \tsh \to \mos $ according to
Eq.\ (7), $ (TV)(B) := \mathrm{tr} \, VF(B) $, we have that, by Eq.\ (8), $ F(B) = T'\chi_{B} $ and consequently
$ F(\Sigma) \subseteq T'\eos $. From $ M \subseteq F(\Sigma) \subseteq T'\eos $ and the fact that $ T'\eos $ is
a convex set of pairwise coexistent effects, we obtain that the convex hull of $M$ is also pairwise coexistent.

Because $ \eh = [0,I] = \frac{1}{2}([-I,I] + I) $ is convex and $ \sigma $-closed (as a consequence of
the Banach-Alaoglu theorem, even $ \sigma $-compact), the $ \sigma $-weak closure of the convex hull of a pairwise
coexistent set of effects is a subset of $ \eh $; by the theorem now, this inclusion is proper. The proof consists in
showing that some non-coexistent effects can physically not be approximated by coexistent ones arbitrarily well.

\begin{theorem}
The convex hull of a set $M$ of pairwise coexistent effects is never $ \sigma $-dense in $ \eh $ (provided that
$ \dim \hi \geq 2 $).
\end{theorem}

\begin{proof} Let $ \fii,\psi \in \hi $, $ \|\fii\| = 1 $, $ \|\psi\| = 1 $, $ \langle \fii|\psi \rangle  = 0 $, and
$ \chi := \frac{1}{\sqrt{2}} (\fii + \psi) $; consider the effects $ P_{\fii} = |\fii\rangle \langle\fii| $ and
$ P_{\chi} = |\chi\rangle \langle\chi| $. Assume that $ \mathrm{conv} \, M $, the convex hull of $M$, were
$ \sigma $-dense in $ \eh $. Then every $ \sigma $-neighborhood of $ P_{\fii} $ and $ P_{\chi} $,
respectively, would contain an element of $ \mathrm{conv} \, M $. Choosing the neighborhoods
\begin{eqnarray}
U(P_{\fii};W_1,W_2,W_3;\varepsilon) & = & \{ A \in \bsh \, | \, |\mathrm{tr} \, W_iA - \mathrm{tr} \, W_iP_{\fii}|
                                      <   \varepsilon, \ i = 1,2,3 \}                              \nonumber        \\
U(P_{\chi};W_1,W_2,W_3;\varepsilon) & = & \{ A \in \bsh \, | \, |\mathrm{tr} \, W_iA - \mathrm{tr} \, W_iP_{\chi}|
                                      <   \varepsilon, \ i = 1,2,3 \}                              \nonumber        \\
\end{eqnarray}
(cf.\ (14)) where $ W_1 = P_{\fii} $, $ W_2 = P_{\psi} $, $ W_3 = P_{\chi} $, and $ \varepsilon = \frac{1}{64} $,
there would exist effects $ G_1,G_2 \in \mathrm{conv} \, M $ such that
$ G_1 \in U(P_{\fii};W_1,W_2,W_3;\varepsilon) $ and
$ G_2 \in $ \linebreak $ U(P_{\chi};W_1,W_2,W_3;\varepsilon) $. By Lemma 5,
$ G_1 $ and $ G_2 $ are coexistent, i.e.,
\begin{eqnarray*}
G_1 & = & A_1 + A_0   \\
G_2 & = & A_2 + A_0
\end{eqnarray*}
where $ A_1,A_2,A_0 \in \eh $ and $ A_1 + A_2 + A_0 \leq I $. Hence, the inequalities
\begin{eqnarray}
|\mathrm{tr} \, W_i(A_1 + A_0) - \mathrm{tr} \, W_iP_{\fii}| & < & \varepsilon     \nonumber\\
|\mathrm{tr} \, W_i(A_2 + A_0) - \mathrm{tr} \, W_iP_{\chi}| & < & \varepsilon     \nonumber\\
\end{eqnarray}
hold for $ i = 1,2,3 $.

The inequalities (21) are equivalent to
\begin{eqnarray}
\mathrm{tr} \, W_iP_{\fii} - \varepsilon   & < &   \mathrm{tr} \, W_i(A_1 + A_0)
                                         \ \ < \ \ \mathrm{tr} \, W_iP_{\fii} + \varepsilon     \nonumber\\
\mathrm{tr} \, W_iP_{\chi} - \varepsilon   & < &   \mathrm{tr} \, W_i(A_2 + A_0)
                                         \ \ < \ \ \mathrm{tr} \, W_iP_{\chi} + \varepsilon.    \nonumber\\
\end{eqnarray}
The addition of the inequalities (22) yields, in particular,
\[
\mathrm{tr} \, W_i(A_1 + A_2 + 2A_0) > \mathrm{tr} \, W_i(P_{\fii} + P_{\chi}) - 2\varepsilon
\]
which implies
\[
\mathrm{tr} \, W_iA_0 > \mathrm{tr} \, W_i(P_{\fii} + P_{\chi}) - \mathrm{tr} \, W_i(A_1 + A_2 + A_0) - 2\varepsilon.
\]
Since $ \mathrm{tr} \, W_i(A_1 + A_2 + A_0) \leq 1 $, it follows that
\begin{equation}
\mathrm{tr} \, W_iA_0 > \mathrm{tr} \, W_i(P_{\fii} + P_{\chi}) - 1 - 2\varepsilon.
\end{equation}
From (22) we further obtain
\begin{eqnarray}
\mathrm{tr} \, W_iA_0 & < & \mathrm{tr} \, W_iP_{\fii} + \varepsilon     \nonumber\\
\mathrm{tr} \, W_iA_0 & < & \mathrm{tr} \, W_iP_{\chi} + \varepsilon.    \nonumber\\
\end{eqnarray}
Inserting $ W_1 = P_{\fii} $, $ W_2 = P_{\psi} $, and $ W_3 = P_{\chi} $ into (23) and (24), we conclude that
\begin{eqnarray}
\frac{1}{2} - 2\varepsilon \ \ < \ \ \mathrm{tr} \, P_{\fii}A_0 & < & \frac{1}{2} + \varepsilon    \\
                                     \mathrm{tr} \, P_{\psi}A_0 & < & \varepsilon                  \\
\frac{1}{2} - 2\varepsilon \ \ < \ \ \mathrm{tr} \, P_{\chi}A_0 & < & \frac{1}{2} + \varepsilon.
\end{eqnarray}

Now, $ (\fii,\psi) \mapsto \langle \fii|A_0\psi \rangle $ is a positive sesquilinear functional. From the
Cauchy-Schwarz inequality and (25) as well as (26) it follows that
\begin{equation}
|\langle \fii|A_0\psi \rangle|^2 \leq \langle \fii|A_0\fii \rangle \langle \psi|A_0\psi \rangle
                                   <  \left( \frac{1}{2} + \varepsilon \right) \varepsilon
                                   =  \frac{1}{2}\varepsilon + \varepsilon^2
                                   <  \varepsilon
\end{equation}
where also $ \varepsilon = \frac{1}{64} < \frac{1}{2} $ has been taken into account. From (25) and (26) again and
from (28) we obtain
\begin{eqnarray*}
\mathrm{tr} \, P_{\chi}A_0 & = &             \langle \chi|A_0\chi \rangle
                             =   \frac{1}{2} \langle \fii + \psi|A_0(\fii + \psi) \rangle                           \\
                           & = & \frac{1}{2} \langle \fii|A_0\fii \rangle + \frac{1}{2} \langle \psi|A_0\psi \rangle
                                                                     + \mathrm{Re} \, \langle \fii|A_0\psi \rangle  \\
                           & < & \frac{1}{4} + \frac{\varepsilon}{2} + \frac{\varepsilon}{2} + \sqrt{\varepsilon}   \\
                           & = & \frac{1}{4} + \varepsilon + \sqrt{\varepsilon}.
\end{eqnarray*}
The last statement $ \mathrm{tr} \, P_{\chi}A_0 < \frac{1}{4} + \varepsilon + \sqrt{\varepsilon} = \frac{25}{64} $
sharpens the right-hand inequality of (27) but contradicts the left which says
$ \mathrm{tr} \, P_{\chi}A_0 > \frac{1}{2} - 2\varepsilon = \frac{30}{64} $. Hence,
$ \mathrm{conv} \, M $ is not $ \sigma $-dense in $ \eh $.
\end{proof}

As a set of effects being pairwise coexistent, $ \mathrm{conv} \, M $ cannot be all of $ \eh $. One could consider
it to be obvious that the $ \sigma $-closure of the set $ \mathrm{conv} \, M $ of pairwise coexistent effects
can also not be equal to the set $ \eh $ which contains pairs of non-coexistent effects as well. But at first glance,
mathematically, this is not obvious since $ \sigma(\bsh,\tsh) $ is a very weak topology. For instance,
if $ \dim \hi = \infty $, the unit sphere $ \{ A \in \bsh \, | \, \|A\| = 1 \} $ is $ \sigma $-dense in the unit ball
$ \{ A \in \bsh \, | \, \|A\| \leq 1 \} $; more generally, the unit sphere of an infinite-dimensional normed space
$ \mathcal{V} $ is $ \sigma(\mathcal{V},\mathcal{V}') $-dense in the unit ball of $ \mathcal{V} $, and the unit sphere
of the dual $ \mathcal{V}' $ is $ \sigma(\mathcal{V}',\mathcal{V}) $-dense in the unit ball of $ \mathcal{V}' $.

Recalling that the physical approximation of an effect by another effect can be described by the $ \sigma $-topology,
one can, from the physical point of view, expect that the $ \sigma $-closure of a pairwise coexistent set of effects
is also pairwise coexistent. In fact, this expectation can be proved on the basis of the second countability and compactness of $ \eh $ in the $ \sigma $-topology. Theorem 5 is then a consequence of Lemma 5 and that result.

\begin{cor}
Let $ F \! : \Sigma \to \eh $ be an observable on $ \os $. Then the convex hull of $ F(\Sigma) $ is never
$ \sigma $-dense in $ \eh $ ($ \dim \hi \geq 2 $).
\end{cor}

If the observable $F$ is statistically complete, the linear hull of $ F(\Sigma) $ is, according to Lemma 1,
$ \sigma $-dense in $ \bsh $. This does not imply that the intersection of $ \mathrm{lin} \, F(\Sigma) $
with the unit ball $ [-I,I] $ of $ \bsh $ is $ \sigma $-dense in the unit ball (see \cite{dix48,sin92}). Since
$ [-I,I] = 2[0,I] - I = 2\eh - I $, $ \mathrm{lin} \, F(\Sigma) \cap \eh $ need not be
$ \sigma $-dense in $ \eh $; however, $ \mathrm{lin} \, F(\Sigma) \cap \eh $ can be
$ \sigma $-dense in $ \eh $, as the statistically complete observable constructed in the proof
of Theorem 1 shows. But the convex hull of $ F(\Sigma) $ cannot be $ \sigma $-dense in $ \eh $.

For a statistical map $T$ it is clear by Lemma 3 that $ T'\eos \subseteq \eh $; according to part (a) of Theorem 3,
$ T'\eos $ is a proper subset of $ \eh $. The following conclusion from Theorem 5 sharpens this result.

\begin{cor}
Let $ T \! : \tsh \to \mos $ be a statistical map. Then $ T'\eos $ \linebreak $ \subseteq \eh $, but
$ T'\eos $ cannot be $ \sigma $-dense in $ \eh $ ($ \dim \hi \geq 2 $).
\end{cor}

In particular, this result holds true for a semiclassical representation on $ \os $. Accordingly,
the function $f$ in inequality (15) can, for an effect $ A \in \eh $, in general not be chosen as a classical effect
$ f \in \eos $. That is, using the terminology introduced after Theorem 4, approximate classical representations of quantum mechanics on $ (\Omega,\Sigma) $ do not exist. This statement is a sharpening of statement (b) of Theorem 3.

Theorem 5, its corollaries, and the implication on the nonexistence of approximate classical representations
are main results of this paper.

The weak operator topology on $ \bsh $ is just the topology $ \sigma(\bsh,\partial_e\sh) $ where
$ \partial_e\sh = \{ P_{\fii} \, | \, \fii \in \hi, \ \|\fii\| = 1 \} $ is the extreme boundary of
$ \sh $; this topology is even weaker than the $ \sigma $-topology, but on norm-bounded subsets of
$ \bsh $ the two topologies coincide. Thus, the convex hull of a pairwise coexistent set of effects can,
in the weak operator topology, also not be dense in $ \eh $. In fact, the neighborhoods (20) with
$ W_1 = P_{\fii} $ etc.\ which play a crucial role in the proof of Theorem 5, are open neighborhoods
w.r.t.\ the weak operator topology as well.

Finally, for a statistical map $T$, we investigate the relation between the set $ T'\eos $ and the range
of the corresponding observable. To that end, we need the following lemma which is, for the convex set
$ \eos $, an analog of the Krein-Milman theorem. Note that $ \eos $  is norm- as well as $ \sigma(\fos,\mos) $-closed, but in general neither norm- nor $ \sigma(\fos,\mos) $-compact; recall that the characteristic functions
$ \chi_B $, $ B \in \Sigma $, are the extreme points of $ \eos $. It is well known that the linear hull
of the characteristic functions is norm-dense in $ \fos $. Their convex hull is norm-dense in
$ \eos $, as is stated now.

\begin{lemma}
The convex hull of the characteristic functions is norm-dense in $ \eos $, i.e.,
\[
\eos = \overline{\mathrm{conv} \, \partial_e\eos}^{\n},
\]
$ \partial_{e}\eos $ denoting the extreme boundary of $ \eos $.
\end{lemma}

\begin{proof} The statement
$ \overline{\mathrm{conv} \, \{ \chi_B \, | \, B \in \Sigma \}}^{\n} \subseteq \eos $ is obvious. To prove
the converse, consider first a function $ g \in \eos $ with finitely many values. Such a function can be written as
$ g = \sum^{n}_{i=1} \alpha_{i} \chi_{B_i} $ where $ 0 \leq \alpha_1 < \alpha_2 < \ldots < \alpha_n \leq 1 $ and
the sets $ B_i \in \Sigma $ form a disjoint decomposition of $ \Omega $. The equality
\begin{eqnarray*}
g   =   \sum_{i=1}^{n} \alpha_i \chi_{B_i}
  & = & \alpha_1 \chi_{B_1 \cup \ldots \cup B_n} + (\alpha_2 - \alpha_1) \chi_{B_2 \cup \ldots \cup B_n}
                                                 + (\alpha_3 - \alpha_2) \chi_{B_3 \cup \ldots \cup B_n} \\
  &   &   + \ldots + (\alpha_n - \alpha_{n-1}) \chi_{B_n} + (1-\alpha_n) \chi_{\emptyset}
\end{eqnarray*}
shows that $ g \in \mathrm{conv} \, \{ \chi_B \, | \, B \in \Sigma \} $. Since an arbitrary function
$ f \in \eos $ can be approximated uniformly by simple functions $ g \in \eos $, we obtain
$ f \in \overline{\mathrm{conv} \, \{ \chi_B \, | \, B \in \Sigma \} }^{\n}$. Hence,
$ \eos = \overline{\mathrm{conv} \, \partial_e\eos}^{\n}$.
\end{proof}

Because $ \eos $ is $ \sigma(\fos,\mos) $-closed, Lemma 6 entails that also
$ \eos = \overline{\mathrm{conv} \, \partial_e\eos}^{\sigma(\fos,\mos)} $. A further consequence
of the lemma is worth mentioning. It is well known that the set $ \eh $ is $ \sigma $-compact and
that its extreme points are the orthogonal projections of $ \hi $; thus, by the Krein-Milman theorem,
$ \eh = \overline{\mathrm{conv} \, \partial_e\eh}^{\sigma} $. From Lemma 6 and the spectral theorem it follows
that the stronger statement $ \eh = \overline{\mathrm{conv} \, \partial_e \eh}^{\n} $ holds. Moreover,
on the basis of this result one can show that already the extreme boundary $ \partial_e\eh $ is,
in the weak operator topology, dense in $ \eh $ \cite{dav76}. Summarizing,
$ \eh = \overline{\mathrm{conv} \, \partial_e\eh}^{\sigma} = \overline{\mathrm{conv} \, \partial_e \eh}^{\n}
= \overline{\partial_e \eh}^{\sigma(\bsh,\partial_e\sh)} = \overline{\partial_e \eh}^{\sigma} $.

\begin{theorem}
Let $T$ be a statistical map and $F$ the observable that corresponds to $T$ according to Theorem 2. Then
\begin{equation}
          \mathrm{conv} \, F(\Sigma) \subseteq T'\eos \subseteq \overline{\mathrm{conv} \, F(\Sigma)}^{\n} \subseteq
\overline{\mathrm{conv} \, F(\Sigma)}^{\sigma} \subset \eh.
\end{equation}
\end{theorem}

\begin{proof} By Eq.\ (8) and the convexity of the set $ T'\eos $, we have
\begin{equation}
\mathrm{conv} \, F(\Sigma) \subseteq T'\eos.
\end{equation}
The linearity of $T'$ and (8) imply that
\begin{equation}
T'(\mathrm{conv} \, \partial_e\eos) = \mathrm{conv} \, T'\partial_e\eos = \mathrm{conv} \, F(\Sigma).
\end{equation}
Using Lemma 6 and the norm continuity of $T'$, we obtain from (31) that
\begin{eqnarray}
T'\eos & = & T' \Big( \overline{\mathrm{conv} \, \partial_e\eos}^{\n} \Big)
       \subseteq      \overline{T'(\mathrm{conv} \, \partial_e\eos)}^{\n}      \nonumber\\
       & = &          \overline{\mathrm{conv} \, F(\Sigma)}^{\n}.              \nonumber\\
\end{eqnarray}
The inclusions (30) and (32) entail
\[
\mathrm{conv} \, F(\Sigma) \subseteq T'\eos \subseteq \overline{\mathrm{conv} \, F(\Sigma)}^{\n},
\]
which concludes the proof since the last two inclusions of (29) are obvious.
\end{proof}

We add a remark. In the case that $ \fos = (\mos)' $, $ \eos $ and $ T'\eos $ are compact
in the respective weak-* topologies; as a consequence, $ T'\eos $ is $ \sigma $-closed,
and the inclusion chain (29) reads
\[
          \mathrm{conv} \, F(\Sigma) \subseteq T'\eos = \overline{\mathrm{conv} \, F(\Sigma)}^{\n}
        = \overline{\mathrm{conv} \, F(\Sigma)}^{\sigma} \subset \eh.
\]

\section{Quasi-Probability Representations}

Since classical and approximate classical representations of quantum mechanics do not exist, the conditions
in Definition 2 must be weakened. Recently, Ferrie, Morris, and Emerson \cite{fer10} introduced the following generalized concept (also cf.\ \cite{fer11,fer08;09}).

\begin{defin}
A \emph{quasi-probability representation of quantum mechanics on \linebreak $ \os $}, the latter again being
a measurable space, is a pair $ (T,S) $ of linear maps $ T \! : \tsh \to \mos $ and $ S \! : \bsh \to \fos $ such that
\begin{enumerate}
\item[(i)] for all $ W \in \sh $, $ (TW)(\Omega) = 1 $
\item[(ii)] $T$ is bounded
\item[(iii)] for all $ W \in \sh $ and all $ A \in \eh $,
\[
\mathrm{tr} \, WA = \int SA \, \mathrm{d}(TW).
\]
\end{enumerate}
\end{defin}

Taking account of $ \nu(\Omega) = \int \chi_{\Omega} \, \mathrm{d}\nu = \langle \nu,\chi_{\Omega} \rangle $ for
$ \nu \in \mos $, we see that condition (i) of the definition means that the set $ \sh $ is mapped into the hyperplane
with the equation $ \langle \nu,\chi_{\Omega} \rangle = 1 $. Similar as in the case of a classical representation
according to Definition~2, $ SI = \chi_{\Omega} $ need not be satisfied for a quasi-probability representation
(cf.\ the remark after the proof of Lemma 4). However, if $ SI = \chi_{\Omega} $ is assumed, then condition (i)
is implied.---By Theorem 3, there is an immediate conclusion from Definition 3 \cite{fer11,fer10}.

\begin{cor}
For a quasi-probability representation $ (T,S) $ on $ \os $, at least one of the following two statements
must be true (provided that $ \dim \hi \geq 2 $):
\begin{enumerate}
\item[(i)] $T$ is not positive, i.e., there exist a state $ W \in \sh $ and a set $ B \in \Sigma $ such that
$ (TW)(B) < 0 $
\item[(ii)] there exist an effect $ A \in \eh $ and a point $ \omega \in \Omega $ such that
$ (SA)(\omega) \notin [0,1] $.
\end{enumerate}
\end{cor}

The next lemma collects some properties of quasi-probability representations $ (T,S) $ and some statements
on them. Note that the dual (adjoint) map $S'$ is, like $T'$, understood w.r.t.\ the dualities
$ \langle \tsh,\bsh \rangle $ and $ \langle \mos,\fos \rangle $.

\begin{lemma} Let $ (T,S) $ be a quasi-probability representation on $ (\Omega,\Sigma) $.
\begin{enumerate}
\item[(a)] For all $ V \in \tsh $,
\begin{equation}
\mathrm{tr} \, V = (TV)(\Omega)
\end{equation}
holds; equivalently, $ T'\chi_{\Omega} = I $. Furthermore, for all $ V \in \tsh $  and all $ A \in \bsh $,
\begin{equation}
\mathrm{tr} \, VA = \int SA \, \mathrm{d}(TV).
\end{equation}
The maps $T$ and $S$ are injective; the map $ T' \! : \fos \to \bsh $ is surjective, and the restriction
$ T'|_{R(S)} $ of $T'$ to the range $ R(S) $ of $S$ equals $ S^{-1} \! : R(S) \to \bsh $.
\item[(b)] For $T$ being bijective, it is necessary that $ \mos $ is norm-separable and that
$ (\mos)' = \fos $. If $T$ is bijective, then $ SI = \chi_{\Omega} $, and $S$ is also
bijective and bounded. Moreover, $ T' = S^{-1} $, $ S' \! : \mos \to \tsh $ exists, and $ S' = T^{-1} $.
\item[(c)] If $T$ is bijective, at least one of the maps $T$ and $S$ is not positive
($ \dim \hi \geq 2) $. Equivalently, if $T$ is positive and bijective (i.e., $T$ is a bijective
semiclassical representation), then there exist an effect $ A \in \eh $ and a point $ \omega \in \Omega $
such that $ (SA)(\omega) < 0 $.
\end{enumerate}
\end{lemma}

\begin{proof} Using the representation $ V = \alpha_1W_1 - \alpha_2W_2 $ for $ V \in \tsh $ where
$ \alpha_1,\alpha_2 \in \mathbb{R} $ and $ W_1,W_2 \in \sh $, we obtain
$ \mathrm{tr} \, V = \alpha_1 - \alpha_2 $ and from condition (i) of Definition 3 that
\[
(TV)(\Omega) = (\alpha_1TW_1 - \alpha_2TW_2)(\Omega) = \alpha_1(TW_1)(\Omega) - \alpha_2(TW_2)(\Omega)
             = \alpha_1 - \alpha_2;
\]
hence, $ \mathrm{tr} \, V = (TV)(\Omega) $. Writing
$ (TV)(\Omega) = \int \chi_{\Omega} \, \mathrm{d}(TV) = \mathrm{tr} \, V(T'\chi_{\Omega}) $,
we conclude from $ \mathrm{tr} \, V = \mathrm{tr} \, V(T'\chi_{\Omega}) $ for all $ V \in \tsh $ that
$ T'\chi_{\Omega} = I $. Conversely, $ T'\chi_{\Omega} = I $ implies $ \mathrm{tr} \, V = (TV)(\Omega) $. Eq.\ (34)
and the statement on the injectivity of $T$ and $S$ are proved the same way
as the analogous statements of Lemma 4. Finally, from
\[
\langle V,A \rangle = \langle TV,SA \rangle = \langle V,T'SA \rangle
\]
for all $ V \in \tsh $ and all $ A \in \bsh $ it follows that $ T'SA = A $ for all $ A \in \bsh $, that is,
$S$ is injective, $T'$ is surjective, and $ T'|_{R(S)} = S^{-1} $.

Next, assume $T$ is bijective. Since $T$ is a norm-continuous linear map between Banach spaces,
$ T^{-1} $ is also norm-continuous. Thus, $T$ is a homeomorphism w.r.t.\ the norm topologies of
$ \tsh $ and $ \mos $, and the norm-separability of $ \tsh $ (which is a consequence of the separability of
$ \hi $) entails the separability of $ \mos $. Further, the assumption implies that the Banach-space adjoint map
$ T^* \! : (\mos)' \to \bsh $ is bijective. Because its restriction $ T' \! : \fos \to \bsh $ is surjective
according to part (a) of the theorem, it follows that $ (\mos)' = \fos $.

From Eqs.\ (34) and (33) we obtain
\[
\int SI \, \mathrm{d}(TV) = \mathrm{tr} \, VI = \mathrm{tr} \, V = (TV)(\Omega)
                          = \int \chi_{\Omega} \, \mathrm{d}(TV).
\]
If $T$ is surjective, we conclude from $ \int SI \, \mathrm{d}(TV) = \int \chi_{\Omega} \, \mathrm{d}(TV) $ that
$ SI = \chi_{\Omega} $.

Continuing with the assumption that $T$ is bijective, let $ \nu \in \mos $ and $ A \in \bsh $ be arbitrary. Then
$ \nu = TV $ for some $ V \in \tsh $, and from
\[
\langle \nu,SA \rangle = \langle TV,SA \rangle = \langle V,A \rangle = \langle T^{-1}\nu,A \rangle
\]
it follows that $S'$ exists and that $ S' = T^{-1} $. Now, from the equality chain
\[
\langle \nu,ST'f \rangle = \langle S'\nu,T'f \rangle = \langle T^{-1}\nu,T'f \rangle = \langle TT^{-1}\nu,f \rangle
                         = \langle \nu,f \rangle
\]
where $ \nu \in \mos $ and $ f \in \fos $, we obtain $ ST'f = f $, i.e., $S$ is surjective. Since $S$ is also injective, it is bijective; the statement $ T'|_{R(S)} = S^{-1} $ then implies $ T' = S^{-1} $. From
$ T' = S^{-1} $ or $ S' = T^{-1} $ we conclude that $S$ is bounded.

Finally, suppose $T$ is positive and bijective, i.e., $T$ is a bijective semiclassical representation. As
was pointed out in Section 2, $ T\sh $ is a proper subset of $ \sos $; equivalently, the image of the positive cone of
$ \tsh $ under $T$ is a proper subset of the positive cone of $ \mos $. Hence, there exists a positive measure
$ \nu \in \mos $ such that $ \nu = TV $ and $ V \in \tsh $ is not positive. As a consequence, there exists
a positive operator $ A \in \bsh $ such that $ \langle V,A \rangle < 0 $. From
\[
\langle \nu,SA \rangle = \langle TV,SA \rangle  = \langle V,A \rangle < 0
\]
it now follows that $SA$ is not positive. Since $ A \geq 0 $, the map $S$ is not positive.
\end{proof}

In general, the space $ \mos $ is not separable and its dual is larger than $ \fos $. However,
the necessary conditions of part (b) of the lemma can be met. Namely, let $ \Omega $ be a countable set
and $ \Sigma $ the power set of $ \Omega $. If $ \Omega $ is infinite, e.g., $ \Omega = \mathbb{N} $, then
$ \mos $ and $ \fos $ are norm- and order-isomorphic to the real sequence spaces
$ l^1_{\mathbb{R}} $ and $ l^{\infty}_{\mathbb{R}} $; so $ \mos $ is separable and
$ (\mos)' = \fos $. If $ \Omega $ is even finite, e.g., $ \Omega = \{ 1,\ldots,N \} $, then
$ \mos $ and $ \fos $ are isomorphic to the spaces $ (\mathbb{R}^N,\|\, . \,\|_1) $ and
$ (\mathbb{R}^N,\|\, . \,\|_{\infty}) $, $ \|\, . \,\|_1 $ being the sum norm and
$ \|\, . \,\|_{\infty} $ the maximum norm; both spaces are separable and each is the dual of the other.---The statements of part (c) sharpen the statement of Corollary 3 and confirm the statements of \cite{fer11,fer08;09,fer10}
on the ``necessity of negativity in quantum theory,'' however, under the additional assumption that $T$ is bijective.

Two further concepts utilized by Ferrie, Morris, and Emerson in the context of quasi-probability representations
are those of a frame and a dual frame \cite{fer11,fer08;09,fer10}. Before presenting their ideas,
we have to specify the notions of an operator-valued measure and of the
$ \sigma $-weak integral of a function w.r.t.\ such a measure. An \emph{operator-valued measure on
$ (\Omega,\Sigma) $}, the latter being a measurable space, is a map $ F \! : \Sigma \to \bsh $ satisfying (i)
$ F(\emptyset) = 0 $ and (ii) $ F(\bigcup_{i=1}^{\infty} B_i) = \sum_{i=1}^{\infty} F(B_i) $ where the sets
$ B_i \in \Sigma $ are mutually disjoint and the sum converges $ \sigma $-weakly (note that in general neither
$F$ is positive nor does it satisfy $ F(\Omega) = I $). In the context of our definition of the
$ \sigma $-weak integral and for the proof of Theorem 7 below, we have to investigate the map
$ V \mapsto \nu_V := \mathrm{tr} \, VF( \, . \, ) $. Clearly, $ \nu_V $ is a $ \sigma $-additive measure on
$ (\Omega,\Sigma) $, as such it is bounded, and $ T \! : \tsh \to \mos $, $ TV := \nu_V $, is linear. To understand
that the boundedness of $T$, as claimed in the following lemma, is not obvious, recall that $ \|\nu_V\| $ is
the total-variation norm of $ \nu_V $ and observe that
\begin{eqnarray*}
\|\nu_V\| &   =  & |\nu_V|(\Omega) = \sup \sum_{i=1}^{n} |\nu_V(B_i)|
              =     \sup \sum_{i=1}^{n} |\mathrm{tr} \, VF(B_i)|  \\       \\
          & \leq &  \|V\|_{\mathrm{tr}} \, \sup \sum_{i=1}^{n} \|F(B_i\|
\end{eqnarray*}
where the suprema are taken over all finite disjoint decompositions of $ \Omega $ into sets
$ B_1,\ldots,B_n \in \Sigma $ and the last supremum can be infinity.---The lemma now generalizes
the main statement of Theorem 2.

\begin{lemma}
Let $F$ be an operator-valued measure on $ (\Omega,\Sigma) $. Then the linear map
$ T \! : \tsh \to \mos $, $ TV = \nu_V $, i.e.,
\[
(TV)(B) = \mathrm{tr} \, VF(B),
\]
$ V \in \tsh $, $ B \in \Sigma $, is bounded. Conversely, every bounded linear map
$ T \! : \tsh \to \mos $ is of the form $ TV = \mathrm{tr} \, VF(\, . \,) $
with a uniquely determined operator-valued measure $ F \! : \Sigma \to \bsh $. In addition,
$F$ is bounded, i.e., $ \|F(B)\| \leq c $ for all $ B \in \Sigma $ and some $ c \in \mathbb{R} $.
\end{lemma}

\begin{proof}
Let $ l_B(V) := \mathrm{tr} \, VF(B) = \nu_V(B) $. Then, on the one hand,
$ |l_B(V)| \leq \|F(B)\| \, \|V\|_{\mathrm{tr}} $, and, on the other hand,
$ |l_B(V)| = |\nu_V(B)| \leq \|\nu_V\| $. That is, we have a family of bounded linear functionals $ l_B $ on
$ \tsh $ satisfying $ |l_B(V)| \leq \|\nu_V\| $ for all $ B \in \Sigma $ and each $ V \in \tsh $. Therefore,
by the principle of uniform boundedness, $ \|l_B\| = \|F(B)\| \leq c $ for all $ B \in \Sigma $. Hence,
$ |\nu_V(B)| = |l_B(V)| \leq \|l_B\| \, \|V\|_{\mathrm{tr}} \leq c \|V\|_{\mathrm{tr}} $; in consequence,
\begin{eqnarray*}
\|\nu_V\| &   =  & \sup \sum_{i=1}^{n} |\nu_V(B_i)|
              =    \sup \Bigg( \sum_{\nu_V(B_i) \geq 0} \nu_V(B_i)
                             - \sum_{\nu_V(B_i) < 0} \nu_V(B_i) \Bigg)                          \\
          &   =  & \sup \Bigg( \nu_V \bigg( \bigcup_{\nu_V(B_i) \geq 0} B_i \bigg)
                      + \bigg| \nu_V \bigg( \bigcup_{\nu_V(B_i) < 0} B_i \bigg) \bigg| \Bigg)   \\
          & \leq & 2c\|V\|_{\mathrm{tr}},
          \end{eqnarray*}
the suprema again being taken over all finite disjoint decompositions of $ \Omega $ into sets
$ B_1,\ldots,B_n \in \Sigma $.

Now assume that $ T \! : \tsh \to \mos $ is linear and bounded. Then the dual map $T'$ exists as a map from
$ \fos $ into $ \bsh $. Define $ F(B) := T'\chi_B $, note that $ F(\emptyset) = 0 $, and prove
$ TV = \mathrm{tr} \, VF(\, . \,) $ as well as the $ \sigma $-additivity of $F$ the same way
as it was done in the proof of Theorem 2 for the case of a semiclassical representation
and a statistically complete observable. The uniqueness of $F$ is obvious.
\end{proof}

From Lemma 8 we conclude that, for a function $ f \in \fos $,
$ \left| \int f \, \mathrm{d}\nu_V \right| \leq \|f\| \, \|\nu_V\| \leq C \|f\| \, \|V\|_{\mathrm{tr}} $. Thus,
$ V \mapsto \int f \, \mathrm{d}(\mathrm{tr} \, VF( \, . \, )) $ is a bounded linear functional on
$ \tsh $, and there exists a unique $ A \in \bsh $ such that
$ \int f \, \mathrm{d}(\mathrm{tr} \, VF( \, . \, )) $ $ = \mathrm{tr} \, VA $. We call the operator
$A$ the \emph{$ \sigma $-weak integral of $f$ w.r.t.\ the operator-valued measure $F$} and write
$ A =: \int f \, \mathrm{d}F $. Briefly, if $ f \in \fos $, then
\[
\int f \, \mathrm{d}(\mathrm{tr} \, VF( \, . \, )) = \mathrm{tr} \left( V \int f \, \mathrm{d}F \right)
\]
for all $ V \in \tsh $. We already used the $ \sigma $-weak integral in the context of Theorem 2
for the particular case that $F$ is an observable.

We can now define frames and their duals in the sense of \cite{fer10}.

\begin{defin} \mbox{}
\begin{enumerate}
\item[(a)] A \emph{frame for $ \tsh $ on $ (\Omega,\Sigma) $} is an operator-valued measure
$ F \! : \Sigma \to \bsh $ such that
\begin{enumerate}
\item[(i)] $F$ is normalized, i.e., $ F(\Omega) = I $
\item[(ii)] for any $ W_1,W_2 \in \sh $,
\[
\mathrm{tr} \, W_1F(B) = \mathrm{tr} \, W_2F(B)
\]
for all $ B \in \Sigma $ implies $ W_1 = W_2 $.
\end{enumerate}
\item[(b)] A \emph{dual frame of $F$} is a family of linear (not necessarily bounded) functionals
$ D_{\omega} \! : \bsh \to \mathbb{R} $, $ \omega \in \Omega $, such that
\begin{enumerate}
\item[(i)] for each $ A \in \bsh $, the function $ \omega \mapsto D_{\omega}(A) $ is
$ \Sigma $-measurable and bounded
\item[(ii)] every $ A \in \bsh $ can be represented as the $ \sigma $-weak integral of
$ \omega \mapsto D_{\omega}(A) $ w.r.t.\ the frame $F$, i.e.,
\begin{equation}
A = \int D_{\omega}(A) \, F(\mathrm{d}\omega).
\end{equation}
\end{enumerate}
\end{enumerate}
\end{defin}

A frame in the sense of this definition generalizes the concept of a statistically complete observable. One
easily verifies that a normalized operator-valued measure is a frame if and only if one of the conditions
(ii), (iii) of Lemma 1 is fulfilled. If $F$ is a frame on $ (\Omega,\Sigma) $, then the linear map
$ T \! : \tsh \to \mos $ defined by
\[
(TV)(B) := \mathrm{tr} \, VF(B),
\]
$ V \in \tsh $, $ B \in \Sigma $, satisfies $ (TW)(\Omega) = 1 $ for $ W \in \sh $, is bounded according to Lemma 8,
and is injective. Following \cite{fer10} again, we call such a map $T$ a \emph{frame representation of
$ \tsh $ on $ (\Omega,\Sigma) $}. Conversely, for every frame representation $T$ there exists, again by Lemma 8,
a unique frame such that $ TV = \mathrm{tr} \, VF(\, . \,) $ (cf.\ Theorem 2). In the context of frames,
a frame representation is the canonical generalization of a semiclassical representation. For every
quasi-probability representation $ (T,S) $, $T$ is a frame representation.

Finally, there is a one-one correspondence between the pairs consisting of a frame and one of its duals and
the quasi-probability representations. This result originates from \cite{fer08;09} and in particular
from \cite{fer10} (also cf.\ \cite{fer11}); on the basis of our preparation, we are able to give a rigorous proof.

\begin{theorem}
If $F$ is a frame on $ (\Omega,\Sigma) $ and the family $ D_{\omega} $, $ \omega \in \Omega $, a dual frame of
$F$, then a quasi-probability representation $ (T,S) $ on $ (\Omega,\Sigma) $ is defined by
\begin{equation}
(TV)(B) := \mathrm{tr} \, VF(B), \qquad (SA)(\omega) := D_{\omega}(A)
\end{equation}
where $ V \in \tsh $, $ B \in \Sigma $, $ A \in \bsh $, and $ \omega \in \Omega $. Conversely,
given any quasi-probability representation $ (T,S) $, there exists a unique frame $F$ and exactly one of its duals
$ D_{\omega} $, $ \omega \in \Omega $, such that $ TV = \mathrm{tr} \, VF( \, . \, ) $ and
$ (SA)(\omega) = D_{\omega}(A) $.
\end{theorem}

\begin{proof}
Let $F$ and $ D_{\omega} $, $ \omega \in \Omega $, be given and consider the linear maps
$ T \! : \tsh \to \mos $ and $ S \! : \bsh \to \fos $ according to (36). From condition (i) of part (a)
of Definition 4 and Lemma 8 it follows that $T$ satisfies conditions (i) and (ii) of Definition 3. Using that
$ D_{\omega} $, $ \omega \in \Omega $, is dual to $F$, we further obtain
\[
\mathrm{tr} \, VA = \mathrm{tr} \left( V \int D_{\omega}(A) \, F(\mathrm{d}\omega) \right)
                  = \int D_{\omega}(A) \, \mathrm{tr} \, VF(\mathrm{d}\omega)
                  = \int SA \, \mathrm{d}(TV).
\]
Now suppose that a quasi-probability representation $ (T,S) $ is given. Then, by Lemma 8,
$T$ uniquely determines an operator-valued measure $F$ satisfying $ TV = \mathrm{tr} \, VF( \, . \, ) $. Because
$ (TW)(\Omega) = 1 $ for $ W \in \sh $, $F$ is normalized; since, by Lemma 7, $T$ is injective,
condition (ii) of part (a) of Definition 4 is implied. Hence, $F$ is a frame. Finally, from
\[
\mathrm{tr} \, VA = \int SA \, \mathrm{d}(TV) = \int SA \, \mathrm{d}(\mathrm{tr} \, VF( \, . \, ))
                  = \mathrm{tr} \left( V \int SA \, \mathrm{d}F \right)
\]
we conclude that $ A = \int SA \, \mathrm{d}F $; that is, $ A \mapsto (SA)(\omega) $, $ \omega \in \Omega $,
is a dual frame of $F$.
\end{proof}

Observe that in the first part of the proof, condition (ii) in the definition of a frame has not been used. That is,
if for a normalized operator-valued measure $F$ (without the assumption of condition (ii) for a frame) there exists
a dual frame $ D_{\omega} $, $ \omega \in \Omega $, in the sense of part (b) of Definition 4, then $F$ and
$ D_{\omega} $, $ \omega \in \Omega $, define a quasi-probability representation $ (T,S) $, and $F$,
conversely determined by $T$, satisfies condition (ii) of part (a) of Definition 4. Hence,
this condition is necessary for the existence of a dual frame.

We still have to discuss the question for examples of quasi-probability representations. In \cite{fer10}
it is argued that an example of a quasi-probability representation can be based on the representation
of the quantum states by the Wigner distribution functions \cite{wig32,wig71}. To understand this statement,
we have to modify the concept specified by Definition 3. Namely, like in Definition 1, let
$ (\Omega,\Sigma,\lambda) $ be a $ \sigma $-finite measure space, let
$ L^{1}_{\mathbb{R}}(\Omega,\Sigma,\lambda) $ and $ L^{\infty}_{\mathbb{R}}(\Omega,\Sigma,\lambda) $
be the corresponding spaces of real $ L^1 $- and $ L^{\infty} $-functions, respectively,
and define a \emph{quasi-probability representation on $ (\Omega,\Sigma,\lambda) $} to be a pair
$ (\widehat{T},\widehat{S}) $ of linear maps $ \widehat{T} \! : \tsh \to L^{1}_{\mathbb{R}}(\Omega,\Sigma,\lambda) $ and $ \widehat{S} \! : \bsh \to L^{\infty}_{\mathbb{R}}(\Omega,\Sigma,\lambda) $ such that
\begin{enumerate}
\item[(i)] for all $ W \in \sh $, $ \int \widehat{T}W \, \mathrm{d}\lambda = 1 $
\item[(ii)] $ \widehat{T} $ is bounded
\item[(iii)] for all $ W \in \sh $ and all $ A \in \eh $,
\[
\mathrm{tr} \, WA = \int (\widehat{T}W)(\widehat{S}A) \, \mathrm{d}\lambda.
\]
\end{enumerate}
The elements of $ \widehat{T}\sh $ are quasi-probability densities. An advantage of this slightly modified
concept $ (\widehat{T},\widehat{S}) $ is again that $ L^{\infty}_{\mathbb{R}}(\Omega,\Sigma,\lambda) $
is the dual space of $ L^{1}_{\mathbb{R}}(\Omega,\Sigma,\lambda) $.

Actually, as was shown by Pool \cite{poo66}, the Wigner functions are in general only $ L^2 $-functions
on the phase space $ \mathbb{R}^2 $, for instance ($ \mathbb{R}^2 $ being equipped with the Lebesgue measure
$ \lambda^2 $ on its Borel sets $ \Xi(\mathbb{R}^2) $). So strictly speaking they cannot be related to an example
of a quasi-probability representation on $ (\mathbb{R}^2,\Xi(\mathbb{R}^2),\lambda^2) $, they only give rise
to a similar construction. Let $ \widehat{T}(P_{\psi}) $ be the Wigner function of a pure quantum state
$ P_{\psi} = |\psi\rangle \langle\psi| $, $ \psi \in \hi $, $ \|\psi\| = 1 $; define it
according to \cite{poo66}, Definition IV.1 and Proposition IV.2, but for some conventional reason
add a factor $ \frac{1}{\sqrt{2\pi}} $ (and set $ \hbar = 1 $). The definition of the map
$ \widehat{T} $ can canonically be extended to $ \sh $ and moreover to $ \tsh $, and we obtain a bounded linear map
$ \widehat{T} \! : \tsh \to L^{2}_{\mathbb{R}}(\mathbb{R}^2) $, the \emph{Wigner transform}. Now introduce
the \emph{Weyl transform} $ H \! : L^{2}_{\mathbb{R}}(\mathbb{R}^2) \to \bsh $ (originating from \cite{wey27})
according to \cite{poo66}, Proposition V.1, but again involving the additional factor
$ \frac{1}{\sqrt{2\pi}} $. We then have
\[
\langle V,Hf \rangle = \mathrm{tr} \, V(Hf) = \int (\widehat{T}V)f \, \mathrm{d}\lambda^2
                     = \langle \widehat{T}V,f \rangle_2
\]
where $ V \in \tsh $, $ f \in L^{2}_{\mathbb{R}}(\mathbb{R}^2) $, and $ \langle \, . \, , \, . \, \rangle_2 $
denotes the scalar product in $ L^{2}_{\mathbb{R}}(\mathbb{R}^2) $ (cf.\ \cite{sch82}); that is, the Weyl transform is dual
to the Wigner transform, $ H =\widehat{T}' = \widehat{T}^* $. The map $H$ is injective with range
$ \mathcal{B}^{HS}_{s}(\mathcal{H}) $, the space of the self-adjoint Hilbert-Schmidt operators. Defining
$ \widehat{S} := H^{-1} $, it follows that, for all $ V \in \tsh $ and all
$ A \in \mathcal{B}^{HS}_{s}(\mathcal{H}) $,
\begin{equation}
\mathrm{tr} \, VA = \langle V,Hf \rangle = \langle V,\widehat{T}'f \rangle = \langle \widehat{T}V,f \rangle_2
                  = \langle \widehat{T}V,H^{-1}A \rangle_2 = \int (\widehat{T}V)(\widehat{S}A) \, \mathrm{d}\lambda^2
\end{equation}
where $ f \in L^{2}_{\mathbb{R}}(\mathbb{R}^2) $ and $ A = Hf $.

Hence, $ \mathrm{tr} \, WA = \int (\widehat{T}W)(\widehat{S}A) \, \mathrm{d}\lambda^2 $ holds true
for all $ W \in \sh $ and all $ A \in \eh \cap \mathcal{B}^{HS}_{s}(\mathcal{H}) $. Thus,
the pair $ (\widehat{T},\widehat{S}) $ consisting of the Wigner transform
$ \widehat{T} \! : \tsh \to L^{2}_{\mathbb{R}}(\mathbb{R}^2) $ and the inverse Weyl transform
$ \widehat{S} \! : \mathcal{B}^{HS}_{s}(\mathcal{H}) \to L^{2}_{\mathbb{R}}(\mathbb{R}^2) $
satisfies condition (ii) and, up to a certain point, also condition (iii) of the definition
of a quasi-probability representation on $ (\mathbb{R}^2,\Xi(\mathbb{R}^2),\lambda^2) $. With some restrictions,
condition~(i) is satisfied as well.---The space $ \mathcal{B}^{HS}_{s}(\mathcal{H}) $ is, as a subspace of $ \bsh $, not closed; its norm closure is the space of the compact self-adjoint operators, and its $ \sigma $-closure is
$ \bsh $. However, one can show that the set $ \eh \cap \mathcal{B}^{HS}_{s}(\mathcal{H})$ is not $ \sigma $-dense in
$ \eh $ (cf.\ the paragraph after Corollary 1); consequently, an arbitrary effect of $ \eh $ can physically
not be approximated by an effect $ A \in \eh \cap \mathcal{B}^{HS}_{s}(\mathcal{H})$, and the replacement of
$ \eh $ by $ \eh \cap \mathcal{B}^{HS}_{s}(\mathcal{H}) $ in condition (iii) is an essential restriction.

We now take into account that $ \mathcal{B}^{HS}_{s}(\mathcal{H}) $ is a Hilbert space w.r.t.\ the real
Hilbert-Schmidt scalar product $ (A,B) \mapsto \langle A,B \rangle_{HS} := \mathrm{tr} \, AB $. As proved
in \cite{poo66}, the Weyl transform, i.e., $ \sqrt{2\pi}H $, is a unitary map from
$ L^{2}_{\mathbb{R}}(\mathbb{R}^2) $ onto $ \mathcal{B}^{HS}_{s}(\mathcal{H}) $. Defining the linear maps
$ \widetilde{T} := \frac{1}{2\pi}H^{-1} $ and $ \widetilde{S} := H^{-1} $ \vspace{0.6mm} from
$ \mathcal{B}^{HS}_{s}(\mathcal{H}) $ onto $ L^{2}_{\mathbb{R}}(\mathbb{R}^2) $, we obtain that, for all
$ V \in \mathcal{B}^{HS}_{s}(\mathcal{H}) $ and all $ A \in \mathcal{B}^{HS}_{s}(\mathcal{H}) $,
\begin{eqnarray}
\mathrm{tr} \, VA & = & \langle V,A \rangle_{HS} = \langle H\varrho,Hf \rangle_{HS}
                    =   \frac{1}{2\pi}\langle \varrho,f \rangle_2
                    =   \frac{1}{2\pi} \int \varrho f \, \mathrm{d}\lambda^2            \nonumber\\
                  & = & \int (\widetilde{T}V)(\widetilde{S}A) \, \mathrm{d}\lambda^2    \nonumber\\
\end{eqnarray}
where $ \varrho,f \in L^{2}_{\mathbb{R}}(\mathbb{R}^2) $, $ V = H\varrho $, and $ A = Hf $. In particular, if
$ V \in \tsh $, it follows from (37), $ \widehat{S} = H^{-1} = \widetilde{S} $, and (38) that
\[
\mathrm{tr} \, VA = \langle \widehat{T}V,\widehat{S}A \rangle_2 = \langle \widehat{T}V,\widetilde{S}A \rangle_2
                  = \langle \widetilde{T}V,\widetilde{S}A \rangle_2.
\]
Since $ A \in \mathcal{B}^{HS}_{s}(\mathcal{H}) $ is arbitrary and $ \widetilde{S} $ is surjective,
this implies that $ \widehat{T}V = \widetilde{T}V $ for all $ V \in \tsh $. That is,
$ \widetilde{T} $ is an extension of $ \widehat{T} $ on $ \mathcal{B}^{HS}_{s}(\mathcal{H}) $. In fact,
$ \widehat{T} $ is continuous w.r.t.\ the trace norm $ \|\,.\,\|_{\mathrm{tr}} $ as well as the Hilbert-Schmidt norm
$ \|\,.\,\|_{HS} $, and $ \tsh $ is $ \|\,.\,\|_{HS} $-dense in $ \mathcal{B}^{HS}_{s}(\mathcal{H}) $,
so $ \widetilde{T} $ is the unique $ \|\,.\,\|_{HS} $-$ \|\,.\,\|_2 $-continuous extension of the Wigner transform
$ \widehat{T} $ on $ \mathcal{B}^{HS}_{s}(\mathcal{H}) $.

According to some discussion in Section 2, a semiclassical representation $ T \! : \tsh \to \mos $ cannot map
the set $ \sh $ onto the set $ \sos $; equivalently, a semiclassical representation cannot be bijective
with a positive inverse $ T^{-1} $. However, a semiclassical representation can be bijective, in this case
the map $ T^{-1} $ is not positive, i.e., the set $ T^{-1}\sos $ contains $ \sh $ as well as some $ V \in \tsh $
being not comparable with $ 0 \in \tsh $. Assume now that $T$ is a bijective semiclassical representation. As
a consequence, $ \mos $ must be norm-separable (cf.\ part (b) of Lemma 7); as an additional assumption, let
$ (\mos)' = \fos $. Then $ T' = T^* $ is also bijective, and by means of the definition
$ S := (T')^{-1} = (T^{-1})' $, $ S \! : \bsh \to \fos $, we obtain
\begin{equation}
\langle TV,SA \rangle = \langle V,T'(T')^{-1}A \rangle = \langle V,A \rangle
\end{equation}
for all $ V \in \tsh $ and all $ A \in \bsh $. Hence, Eq.\ (34) is fulfilled, and $ (T,S) $ is a quasi-probability
representation where $T$ is a positive bijective linear map and, by part (c) of Lemma 7, $S$ is not.

The same way a quasi-probability representation is obtained if, more generally, $T$ is assumed to be
a bijective frame representation (i.e., $T$ is a bijective bounded linear map $ T \! : \tsh \to \mos $
such that $ (TW)(\Omega) = 1 $ for all $ W \in \sh $). In view of part (b) of Lemma 7, we have achieved
the following result: A bijective frame representation $T$ and a linear map $ S \! : \bsh \to \fos $
constitute a quasi-probability representation $ (T,S) $ if and only if the conditions
$ (\mos)' = \fos $ and $ S = (T')^{-1} $ are satisfied.

Thus, the question of the existence of quasi-probability representations on $ (\Omega,\Sigma) $ can be reduced to
the question of the existence of bijective semiclassical (or frame) representations on $ (\Omega,\Sigma) $ where
$ (\mos)' = \fos $. The latter question has a positive answer, at least for a finite-dimensional Hilbert space,
as we are going to show. Let $ \dim \mathcal{H} = n $, then $ \tsh = \bsh $, $ \dim \bsh = n^2 =: N $,
and the bilinear functional
\begin{equation}
(V,A) \mapsto \langle V,A \rangle = \mathrm{tr} \, VA = \langle V,A \rangle_{HS}
\end{equation}
coincides with the Hilbert-Schmidt scalar product in $ \bsh $. Notice that,
in spite of their set-theoretical equality, the spaces $ \tsh $ and $ \bsh $ are equipped
with different norms which again are different from the Hilbert-Schmidt norm. However,
because of the finite dimension, all the norms are equivalent.

It can be shown that there exist algebraic bases $ F_1,\ldots,F_N $ of $ \bsh $ consisting of positive operators
satisfying $ \sum_{i=1}^{N} F_i = I $ \cite{bus93,hel93}. Such a basis defines a statistically complete observable on
$ (\Omega,\Sigma) $ where $ \Omega := \{ 1,\ldots,N \} $ and $ \Sigma $ is the power set of $ \Omega $. Since
in the case of this measurable space the probability measures can be identified with the stochastic vectors of
$ \mathbb{R}^N $ and the space $ \mos $ with $ \mathbb{R}^N $ equipped with the sum norm, the map
$ T \! : \tsh \to \mathbb{R}^N $,
\begin{equation}
TV := \left( \begin{array}{c} \mathrm{tr} \, VF_1 \\ \vdots \\ \mathrm{tr} \, VF_N \end{array} \right),
\end{equation}
is a bijective semiclassical representation. Furthermore, the space $ \fos $ can be identified with
$ \mathbb{R}^N $ with the maximum norm, and the bilinear functional
\[
(\nu,f) \mapsto \langle \nu,f \rangle = \int f \, \mathrm{d}\nu = \sum_{i=1}^{N} q_ia_i = \langle q,a \rangle_2
\]
coincides with the standard scalar product of $ \mathbb{R}^N $ where $ q \in \mathbb{R}^N $ characterizes the measure
$ \nu \in \mos $ and $ a \in \mathbb{R}^N $ characterizes the function $ f \in \fos $. From
\[
\langle TV,a \rangle_2 = \sum_{i=1}^{N} a_i \mathrm{tr} \, VF_i
                       = \mathrm{tr} \left( V\sum_{i=1}^{N} a_iF_i \right)
                       = \langle V,T'a \rangle_{HS},
\]
$ V \in \bsh $, $ a \in \mathbb{R}^N $, it follows for $ T' \! : \mathbb{R}^N \to \bsh $ that
\begin{equation}
T'a = \sum_{i=1}^{N} a_iF_i.
\end{equation}
Clearly, $T'$ is also bijective. Representing an arbitrary operator $ A \in \bsh $ w.r.t.\ the basis
$ F_1,\ldots,F_N $, $ A = \sum_{i=1}^{N} \alpha_iF_i $, (42) implies for $ (T')^{-1} \! : \bsh \to \mathbb{R}^N $
that
\begin{equation}
(T')^{-1}A = \left( \begin{array}{c} \alpha_1 \\ \vdots \\ \alpha_N \end{array} \right).
\end{equation}
Defining $ S := (T')^{-1} $, we obtain from (41) and (43) that
\begin{equation}
\langle TV,SA \rangle_2 = \sum_{i=1}^{N} \alpha_i\mathrm{tr} \, VF_i
                        = \mathrm{tr} \left( V\sum_{i=1}^{N} \alpha_i F_i \right)
                        = \mathrm{tr} \, VA,
\end{equation}
in accordance with the general result (39). Thus, we have constructed a quasi-probability representation
$ (T,S) $ based a on bijective semiclassical representation. In the equation
\[
\mathrm{tr} \, WA = \langle p,a \rangle_2
\]
for the probability for the occurrence of the effect $ A \in \eh $ in the state $ W \in \sh $ the stochastic vector
$ p = TW $ and $ a = SA \in \mathbb{R}^N $ are explicitly given by formulas (41) and (43).

There are also bases of the $N$-dimensional space $ \bsh $ consisting of not necessarily positive operators
$ F_1,\ldots,F_N $ satisfying $ \sum_{i=1}^{N} F_i = I $. Such a basis defines, in the sense of Definition 4,
a frame on the previous discrete measurable space $ (\Omega,\Sigma) $ and, by Eq.\ (41), a bijective
frame representation $T$. The maps $T'$ and $ S := (T')^{-1} $ are again given by Eqs.\ (42) and (43). Since
the equality chains (39) and (44) hold true as well, $ (T,S) $ is a quasi-probability representation
of more general type than that of the preceding paragraph (cf.\ \cite{fer11,fer08;09,fer10}).

Let $ F_1,\ldots,F_N $ be a basis defining a frame and
$ D_1,\ldots,D_N $ the dual basis w.r.t.\ the scalar product (40), i.e.,
$ \langle F_i,D_j \rangle_{HS} = \delta_{ij} = \mathrm{tr} \, F_iD_j $. The coefficients of
$ A \in \bsh $ in the linear combination $ A = \sum_{i=1}^{N} \alpha_iF_i $ can then be calculated according to
$ \alpha_i = \langle D_i,A \rangle_{HS} = \mathrm{tr} \, AD_i $. Comparing the representation
\[
A = \sum_{i=1}^{N} \langle D_i,A \rangle_{HS} F_i
\]
with Eq.\ (35) and identifying the operators $ D_i $ with the linear functionals
$ A \mapsto \langle D_i,A \rangle_{HS} $, $ i \in \Omega = \{ 1,\ldots,N \} $, we see that the dual basis
$ D_1,\ldots,D_N $ is a dual frame of $ F_1,\ldots,F_N $. In this particular situation the dual frame
is unique. For $S$ we obtain from (43) that
\begin{equation}
SA = \left( \begin{array}{c} \mathrm{tr} \, AD_1 \\ \vdots \\ \mathrm{tr} \, AD_N \end{array} \right).
\end{equation}
This representation of $S$ is a particular case of the second equation of (36)
and is completely analogous to the representation (41) of $T$. Moreover, from
$ \sum_{i=1}^{n} F_i = I $ and $ \mathrm{tr} \, F_iD_j = \delta_{ij} $ it follows that
\[
\mathrm{tr} \, D_j = \mathrm{tr} \left[ \left( \sum_{i=1}^{n} F_i \right) D_j \right]
                   = \sum_{i=1}^{n }\mathrm{tr} \, F_iD_j = \sum_{i=1}^{n } \delta_{ij} = 1;
\]
thus, $ \mathrm{tr} \, D_1 = \ldots = \mathrm{tr} \, D_j = 1 $. As a consequence,
$ SI = \left( \begin{array}{c} 1 \\ \vdots \\ 1 \end{array} \right) $ which corresponds to the general result
stated in part (b) of Lemma 7 that $ SI = \chi_{\Omega} $ if $T$ is bijective.---Using (41) and (45),
Eq.\ (34) can directly be verified by the simple calculation
\begin{equation}
\mathrm{tr} \, VA = \mathrm{tr} \left( V\sum_{i=1}^{N} (\mathrm{tr} \, AD_i)F_i \right)
                  = \sum_{i=1}^{N} (\mathrm{tr} \, VF_i)(\mathrm{tr} \, AD_i)
                  = \langle TV,SA \rangle_2
                  = \int f \, \mathrm{d}\nu
\end{equation}
where $f$ is the function on $ \Omega $ determined by $ SA \in \mathbb{R}^N $ and $ \nu $ the measure on
$ \Sigma $ corresponding to $ TV \in \mathbb{R}^N $. The equality chain (46) is a rewriting of (44). Finally,
according to part (c) of Lemma 7, both $T$ and $S$ cannot be positive linear maps; by Eqs.\ (41) and (45)
it can be seen that, equivalently, both bases $ F_1,\dots,F_N $ and $ D_1,\dots,D_N $ cannot exclusively consist
of positive operators. For a proof of this fact which is adapted to the particular situation, see \cite{hel93}.

Summarizing, in the case of a finite-dimensional Hilbert space it is easy to construct quasi-probability representations $ (T,S) $ of quantum mechanics where $T$ and $S$ are even bijective. Moreover, $T$ can
be chosen to be positive, i.e., $T$ is a semiclassical representation. In the case of
an infinite-dimensional Hilbert space the construction of quasi-probability representations is
not obvious. A bijective semiclassical (or frame) representation induces, under the above duality condition,
a quasi-probability representation, but the existence or construction of an example of the former is,
in the infinite-dimensional case, not obvious as well.

\end{document}